\documentclass[twocolumn,superscriptaddress,amsfont,amssymb,amsmath, showpacs,balancelastpage]{revtex4-1}
\usepackage{CJKutf8}

\usepackage{amsthm}

\usepackage[utf8]{inputenc}
\usepackage{amsmath,amssymb,amsfonts,stmaryrd}
\usepackage{physics}
\usepackage{dsfont}
\usepackage{graphicx}
\usepackage{xcolor}
\usepackage{graphicx}
\usepackage{cancel}
\usepackage{natbib}
\usepackage{color}
\usepackage{tikz}
\usepackage{mathrsfs}
\usepackage{relsize}
\usepackage{multirow}
\usepackage[linktocpage]{hyperref}
\usepackage[all]{xy}
\hypersetup{colorlinks=true,citecolor=blue,linkcolor=blue, urlcolor=blue, breaklinks=true}

\usepackage{float}

\usepackage{bm} 
\usepackage{subfigure} 
\usepackage{environ}
\usepackage{url}
\usepackage[margin=0.75in]{geometry}
\usepackage{ulem}

\usepackage{mathtools}

\newcommand{\Z}{{\mathbb Z}}

\def\U{\mathrm{U}(1)}
\def\nfs{N_{\mathrm{FS}}}
\def\TT{\mathcal{T}}

\NewEnviron{eqs}{%
\begin{equation}\begin{split}
    \BODY
\end{split}\end{equation}}

\definecolor{dkgreen}{rgb}{0,0.5,0}

\theoremstyle{definition}

\theoremstyle{remark}

\newtheorem{thm}{Theorem}

\begin{document}

\begin{CJK*}{UTF8}{bsmi}

\title{Extracting higher central charge from a single wave function}

\author{Ryohei Kobayashi}
\email[E-mail: ]{ryok@umd.edu}
\affiliation{Department of Physics, Condensed Matter Theory Center, and Joint Quantum Institute, University of Maryland, College Park, Maryland 20742, USA}

\author{Taige Wang}
\affiliation{Department of Physics, University of California, Berkeley, California 94720 USA}
\affiliation{Material Science Division, Lawrence Berkeley National Laboratory, Berkeley, CA 94720, USA}

\author{Tomohiro Soejima (副島智大)}
\affiliation{Department of Physics, University of California, Berkeley, California 94720 USA}

\author{Roger S. K. Mong (蒙紹璣)}
\affiliation{Department of Physics and Astronomy,
University of Pittsburgh, Pittsburgh, PA 15260, USA}

\author{Shinsei Ryu}
\affiliation{Department of Physics, Princeton University, Princeton, New Jersey, 08544, USA}

\date{\today}
\begin{abstract}
A (2+1)D topologically ordered phase may or may not have a gappable edge, even if its chiral central charge $c_-$ is vanishing. Recently, it is discovered that a quantity regarded as a ``higher'' version of chiral central charge gives a further obstruction beyond $c_-$ to gapping out the edge. 
In this Letter, we show that the higher central charges can be characterized by the expectation value of the \textit{partial rotation} operator acting on the wavefunction of the topologically ordered state. This allows us to extract the higher central charge from a single wavefunction, which can be evaluated on a quantum computer.
Our characterization of the higher central charge is analytically derived from the modular properties of edge conformal field theory, as well as the numerical results with the $\nu=1/2$ bosonic Laughlin state and the non-Abelian gapped phase of the Kitaev honeycomb model, which corresponds to $\mathrm{U}(1)_2$ and Ising topological order respectively. 
The letter establishes a numerical method to obtain a set of obstructions to the gappable edge of (2+1)D bosonic topological order beyond $c_-$, which enables us to completely determine if a (2+1)D bosonic Abelian topological order has a gappable edge or not.
We also point out that the expectation values of the partial rotation on a single wavefunction put a constraint on the low-energy spectrum of the bulk-boundary system of (2+1)D bosonic topological order, reminiscent of the Lieb-Schultz-Mattis type theorems.
 \end{abstract}

\maketitle

\end{CJK*}




\textit{ Introduction --}
(2+1)D topological phases with bulk energy gap host various intriguing physical phenomena~\cite{wen2004quantum}. 
One of the most striking is the bulk-edge correspondence, 
where the property of the bulk heavily constrains dynamics at its boundary. The most celebrated example is the Integer Quantum Hall effect, where nonzero bulk Chern number implies the presence of gapless charged edge modes~\cite{hatsugai1993}. 
Even without charge conservation, systems with nonzero chiral central charge $c_-$, which signals nonzero \textit{thermal} Hall conductance, has gapless edge modes~\cite{KaneFisher}. 
We have a good theoretical understanding of these quantities through coarse-grained Chern-Simons theory, and we can extract them from microscopic wavefunctions~\cite{Kitaevanyons, Mitchell2018amorphus, Kim2022cminus, Kim2022modular, Zou2022modular, Fan2022cminus, Fan2022QHE, Shiozaki2018antiunitary, Dehghani_2021, Cian_2021}.

In the presence of anyonic excitations, there are properties beyond $c_-$ that enforces the presence of gapless edge modes. In many cases, nontrivial braiding statistics between anyons can present an obstruction to gapping out all anyonic degrees of freedom simultaneously at the boundary~\cite{Kapustin:2010hk, Levin2013edge}. Such phases of matter are said to have an \textit{ungappable} edge. Recently, it is discovered that a quantity called \textit{higher central charge} can partially capture ``ungappability'' of the edge~\cite{Ng2018higher, Ng2020higher}. 
In particular, higher central charges of an Abelian topological order completely determines whether it has an ungappable edge~\cite{kaidi2021higher}. However, so far the quantity has been characterized purely through the topological quantum field theory (TQFT) framework, and a microscopic understanding of higher central charges has been lacking.

In this Letter, we show that the expectation value of the \textit{partial rotation} operator -- rotation operator that acts only on a part of the system --
can be used to reliably extract higher central charges of topologically ordered systems. This is the first proposal that relates the wavefunction of a topological ordered state to its higher central charges, 
and our operational definition even allows its evaluation on a quantum computer.  Our finding is supported by an analytical conformal field theory (CFT) calculation, as well as numerics on the non-Abelian phase of the Kitaev honeycomb model and $\nu=1/2$ bosonic Laughlin state. 
This Letter establishes a general numerical method to obtain obstructions to a gappable edge of a bosonic topological order beyond $c_-$, which enables us to completely determine if bosonic Abelian topological order has a gappable edge. 

\begin{figure}[htbp]
    \centering
    \includegraphics[width = 0.48\textwidth]{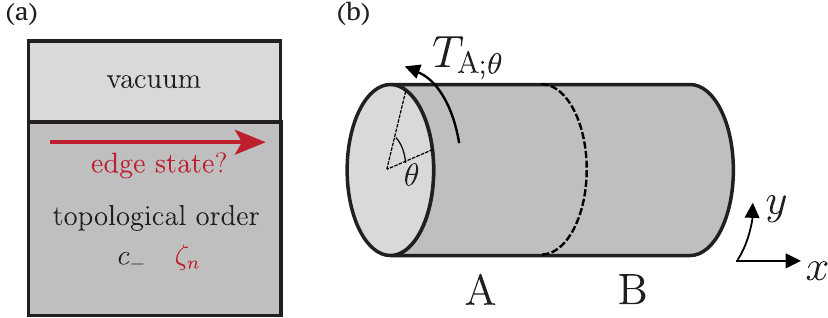}
    \caption{(a) The setup for which we considered the gappability problem. The obstruction can be captured by $c_-$ and higher central charge $\zeta_n$. (b) Schematics of the partial rotation of a cylinder bisected into A and B subsystems.}
    \label{fig:cylinder}    
\end{figure}

\noindent \textit{Definition and properties of higher central charge --} The higher central charges $\zeta_n$ are complex numbers characterizing a topologically ordered state, labeled by a positive integer $n$. $\zeta_n$ can be easily computed from the properties of anyons in the topological order; for a given bosonic (2+1)D topologically order, $\zeta_n$ is defined as the following phase
\begin{align}
    \zeta_n = \frac{\sum_{a} d_a^2\theta_a^n}{\mathopen\big\lvert\sum_{a} d_a^2\theta_a^n\big\rvert}~,
\end{align}
where the sum is over all anyons in the topological order. $d_a$ is quantum dimension, and $\theta_a$ is topological twist of an anyon $a$.
When $n=1$, $\zeta_1$ reduces to the Gauss sum formula for chiral central charge modulo 8, $\zeta_1 = e^{\frac{2\pi i}{8}c_-}$, hence $\zeta_n$ formally provides a generalization of $c_-$. 

Higher central charges put a constraint on the gappability of the edge; it was proven in~\cite{Ng2018higher} that $\zeta_n=1$ for all $n$ such that $\gcd(n, N_{\mathrm{FS}})=1$ give necessary conditions for a gappable edge. Here, $N_{\mathrm{FS}}$ is called the Frobenius-Schur exponent, defined as the smallest positive integer satisfying $\theta_a^{N_{\mathrm{FS}}}=1$ for all anyons $a$. For example, $\mathrm{U}(1)_2\times \mathrm{U}(1)_{-4}$ Chern-Simons theory has $\zeta_3 = -1$, which shows that the topological order has an ungappable edge even though $c_-=0$. 
For (2+1)D bosonic Abelian topological phases one can also derive sufficient conditions: the higher central charges $\{\zeta_n\}$ for $\gcd(n,\frac{N_{\mathrm{FS}}}{\gcd(n,N_{\mathrm{FS}} )})=1$ give both necessary and sufficient conditions for a gappable boundary \cite{kaidi2021higher}.

\textit{ Main result --} 
To extract higher central charges from a single wavefunction, we consider a (2+1)D topological ordered state located on a cylinder. The state on the cylinder is labeled by the anyon $a$, which corresponds to a quasiparticle obtained by shrinking the puncture at the end of the cylinder.
Suppose we have realized a ground state $\ket{\Psi}$ on the cylinder labeled by the trivial anyon $1$.
Let us take a bipartition of the cylinder into the two subsystems labeled by A and B, and write the translation operator for the A subsystem by the angle $\theta$ along the circumference as $T_{\mathrm{A};\theta}$ (see Fig.~\ref{fig:cylinder}). We then find that the following quantity extracts $\zeta_n$,
\begin{align}
\begin{split}
\mathcal{T}_1\left(\frac{2\pi}{n}\right) &:={\bra{\Psi} T_{\mathrm{A};\frac{2\pi}{n}}\ket{\Psi}} 
\propto e^{-2\pi i (\frac{2}{n}+n)\frac{c_-}{24}} \times \sum_a d_a^2\theta_a^n~,
\label{eq:mainresult}
\end{split}
\end{align}
where $\propto$ in this Letter always means being proportional up to a positive real number.
In the special case where $n=1$, the rhs becomes 1 since $\sum_a d_a^2\theta_a\propto e^{\frac{2\pi i}{8}c_-}$, consistent with the fact that the $2\pi$ rotation of the cylinder A gives the identity. For $n > 1$ and $\gcd(n, N_{\mathrm{FS}})=1$, the above rhs becomes proportional to $\zeta_n$ and gives a non-trivial obstruction to gapped boundary beyond $c_-$. Since $c_-$ can be extracted from a single wavefunction \cite{Kim2022cminus, Kim2022modular}, our method allows a complete characterization of all higher central charges. 

For (2+1)D bosonic Abelian topological order one can show that partial rotation, together with topological entanglement entropy~\cite{KitaevTEE, Levin2006Detecting}, fully determines if its edge is gappable. See Supplemental Materials for an explicit algorithm determining gappability. We also note that partial rotation, which is unitary, can be easily evaluated on a quantum computer using methods such as the Hadamard test.

\textit{Analytic derivation --}
Eq.~\eqref{eq:mainresult} can be derived by employing the cut-and-glue approach established in~\cite{KitaevTEE, Qi2012entanglement}, which describes the entanglement spectrum of the A subsystem at long wavelength by that of the (1+1)D CFT on its edges~\footnote{While the argument in~\cite{Qi2012entanglement} is valid for chiral (i.e., holomorphic) edge CFT, one can utilize the cut-and-glue approach for non-chiral edge CFT as well, as demonstrated in Supplemental Materials.}.
Namely, the reduced density matrix for the A subsystem is effectively given by $\rho_{\mathrm{A}} =\rho_{\mathrm{A};l}\otimes \rho_{\mathrm{A};r}$,
where $\rho_{\mathrm{A};l}, \rho_{\mathrm{A};r}$ denote the CFTs on the left and right edges respectively. The left edge lies at the end of the whole cylinder realizing the ground state of CFT;
the right edge of the A subsystem entangled with the B subsystem is described by a thermal density matrix of a perturbed edge CFT~\cite{Haldane2008entanglement}.
The form of the perturbation in the entanglement Hamiltonian is not universal.
In the following, we assume that the entanglement Hamiltonian is that of the unperturbed CFT: $\rho_{\mathrm{A};r} = e^{-\beta_r H_r}$, and check the validity of this assumption with our numerics.

Since the operator $T_{\mathrm{A};\theta}$ acts as the translation of the edge CFT, the partial rotation is expressed as the expectation values of translation operators within the edge CFT as
\begin{align}
\begin{split}
    \mathcal{T}_1\left(\frac{2\pi}{n}\right)  &= \frac{\mathrm{Tr}\big[ e^{i{P}_l\frac{L}{n}}e^{-\frac{\xi_l}{v}H_l} \big] \mathrm{Tr}\big[e^{i{P}_r\frac{L}{n}}e^{-\frac{\xi_r}{v}H_r} \big]}{\mathrm{Tr}\big[ e^{-\frac{\xi_l}{v}H_l} \big] \mathrm{Tr}\big[ e^{-\frac{\xi_r}{v}H_r} \big]} \\
    &= \frac{\chi_1\big(\frac{i\xi_l}{L}+\frac{1}{n}\big) \, \chi_1\big(\frac{i\xi_r}{L}-\frac{1}{n}\big) }{\chi_1\big(\frac{i\xi_l}{L}\big) \, \chi_1\big(\frac{i\xi_r}{L}\big)}
    ,\label{eq:rotasCFT}
    \end{split}
\end{align}
where we introduced the velocity $v$, correlation length $\xi_l = v \beta_l$, $\xi_r = v \beta_r$, and the circumference of the cylinder $L$.
${P}_l$ and ${P}_r$ are translation operators on the left and right edge ${P}_l=-\frac{1}{v}H_l$, ${P}_r=\frac{1}{v}H_r$. $\chi_1(\tau)$ is the CFT character of the trivial sector with modular parameter $\tau$. In our setup where $L \ll \xi_l$, the characters for the left edge are approximated as
\begin{align}
    \chi_1\left(\frac{i\xi_l}{L}\right) \approx e^{\frac{2\pi \xi_l}{L}\frac{c_-}{24}}, \ \chi_1\left(\frac{i\xi_l}{L}+\frac{1}{n}\right) \approx e^{\frac{2\pi \xi_l}{L}\frac{c_-}{24}} e^{-\frac{2\pi i}{n}\frac{c_-}{24}}.
    \end{align}
Meanwhile, the edge CFT at the right edge cutting the system has high temperature $L \gg \xi_r$. These characters can be approximately computed by performing proper modular $S, T$ transformations as~\cite{Shiozaki2017point}
\begin{align}
    \begin{split}
        \chi_1\left(\frac{i\xi_r}{L}\right) &= \sum_a S_{1,a} \chi_a\left(\frac{iL}{\xi_r}\right) \approx \frac{1}{\mathcal{D}}e^{\frac{2\pi L}{\xi_r}\frac{c_-}{24}},
    \end{split}
\\
    \begin{split}
        \chi_1\left(\frac{i\xi_r}{L}-\frac{1}{n}\right) &= \sum_a (ST^n S)_{1,a} \chi_a\left(\frac{iL}{n^2\xi_r}+\frac{1}{n}\right) \\
        &\approx (ST^n S)_{1,1}e^{-\frac{2\pi i}{n}\frac{c_-}{24}} e^{\frac{2\pi L}{n^2\xi_r}\frac{c_-}{24}} \\
        &= \frac{1}{\mathcal{D}^2} e^{-2\pi i(n + \frac{1}{n})\frac{c_-}{24}} e^{\frac{2\pi L}{n^2\xi_r}\frac{c_-}{24}}\sum_{a}d_a^2\theta_a^n,
        \end{split}
\end{align}
 where $n$ is assumed to be small satisfying $n^2\ll L/\xi_r$.
 The sum is over the anyons $a$ that labels the conformal block of the edge CFT, and $\mathcal{D} = \sqrt{\sum_a d_a^2}$ is the total quantum dimension. By combining the above approximations of the characters, $\mathcal{T}_1\left(2\pi/n\right)$ in Eq.~\eqref{eq:rotasCFT} is expressed as Eq.~\eqref{eq:mainresult}.

A similar computation can be performed when the ground state lives in a generic topological sector, 
\begin{align} 
\begin{split}
\mathcal{T}_a\left(\frac{2\pi}{n}\right) \propto  e^{\frac{2\pi i}{n} h_a-2\pi i (\frac{2}{n}+n)\frac{c_-}{24}} \times \zeta_{n,a}~,
\label{eq:fullresulttwisted_main}
\end{split}
\end{align}
where $\mathcal{T}_a\left({2\pi}/{n}\right) := \bra{\Psi_a} T_{\mathrm{A};{2\pi}/{n}}\ket{\Psi_a}$ with $\ket{\Psi_a}$ being the ground state in the topological sector labeled by an anyon $a$. We defined twisted higher central charge
\begin{align}
    \zeta_{n,a}:=\sum_{b}S_{ab}d_b\theta_b^n~,
\end{align} 
which is proportional to $\zeta_n$ when $a=1$.
The derivation of Eq.~\eqref{eq:fullresulttwisted_main} is relegated to Supplemental Materials.

While the definition of the quantity Eq.~\eqref{eq:mainresult} is akin to that of the momentum polarization in the large $n$ limit~\cite{Qi2012momentumpolarization, FQHEDMRG}, we emphasize that the partial rotation by the finite angle $\TT_a(2\pi/n)$ extracts a completely different universal quantity from the momentum polarization. Indeed, the momentum polarization with $n\to\infty$ does not give the higher central charge, which is expressed as 
\begin{equation} \label{eq:momentumpolarization}
    \lim_{n \to \infty} \mathcal{T}_a\left(\frac{2\pi}{n}\right) \propto\exp \left[\frac{2 \pi i }{n} \left(h_a-\frac{c_-}{24}-\frac{c_-}{24} \frac{L^2}{\xi_r^2} \right)\right] .
\end{equation}
Remarkably, while Eq.~\eqref{eq:momentumpolarization} depends on the circumference $L$ and the non-universal correlation length $\xi_r$, Eq.~\eqref{eq:mainresult} solely gives a constant universal value as the combination of $c_-$ and $\zeta_n$. 
In Supplemental Materials, we describe how the behavior of the partial rotation interpolates between higher central charge and momentum polarization.

\begin{figure}
    \centering
    \includegraphics[width=\linewidth]{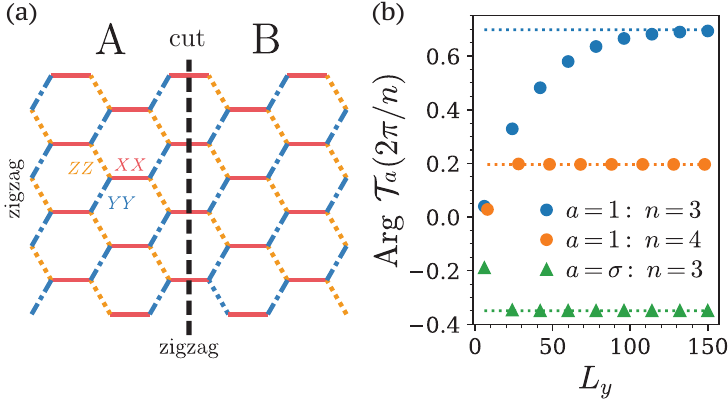}
    \caption{(a) Geometry of the Kitaev model on a cylinder. Red, blue and yellow lines correspond to $X$, $Y$ and $Z$ type Ising interactions, respectively. The lattice is periodic in the $y$ direction, and has the zigzag boundary condition in the $x$ direction. (b) The partial rotations $\mathcal{T}_a(2\pi/n)$ evaluated in the Ising topological phase 
    of the Kitaev model at $n=3,4$. 
    The $\sigma$ sector at $n=4$ is not shown since it evaluates zero. We used $J_x = J_y = J_z = 1, \kappa = 0.1$ for computation.}
    \label{fig:kitaev_main_text}
\end{figure}

\textit{Numerical results --} 
We demonstrate the validity of the formula Eq.~\eqref{eq:mainresult} for two examples: the Ising TQFT realized by the Kitaev honeycomb model, and the $\mathrm{U}(1)_2$ TQFT realized by the $\nu=1/2$ bosonic Laughlin state. Their $\zeta_{n,a}$ and expected values of the partial rotation $\mathcal{T}_a({2\pi}/{n})$ are summarized in Table~\ref{tab:phases}. For some of the $n$'s in a given topological sector, the magnitude of $\mathcal{T}_1$ vanishes. However, this could only occur when $\gcd(n, N_{\mathrm{FS}}) \neq 1$, which therefore does not obscure the examination of whether the topological order has a gappable boundary.

The Kitaev honeycomb model is defined on a honeycomb lattice with a qubit on each vertex, with the Hamiltonian
\begin{align}
\begin{split}
    H &= J_x \sum_{\langle ij\rangle\in \text{R edge}}X_iX_j  + J_y\sum_{\langle ij\rangle\in \text{B edge}}Y_iY_j  \\
     &\quad + J_z\sum_{\langle ij\rangle\in \text{Y edge}}Z_iZ_j + \kappa \sum_{\langle ijk \rangle} X_i Y_j Z_k,
    \end{split}
\end{align}
where the last term is introduced by turning on magnetic field, which realizes the non-Abelian gapped phase~\cite{Kitaevanyons}.
The non-Abelian phase is known to host Ising TQFT with anyons $1$, $\sigma$, $\psi$ with topological twists $\theta_1 = 1$, $\theta_\sigma = e^{2\pi i/16}$, $\theta_\psi = -1$.

\begin{table}[]
\renewcommand*{\arraystretch}{1.5}
\centering
\begin{tabular}{cc|c|c}
\hline
       & sector $a$  & $\zeta_{n,a}$ & $\mathcal{T}_a\left(\frac{2\pi}{n}\right)$ \\ \hline
\multirow{2}{*}{Ising}  & Trivial $1$ & $e^{\frac{2\pi i}{16}},e^{\frac{2\pi i}{16}},e^{\frac{6\pi i}{16}},e^{\frac{4\pi i}{16}}$   &  $1,1,e^{\frac{2\pi i}{9}}, e^{\frac{\pi i}{16}}$  \\ 
                        & $\sigma$     & $1,0,1,0$  &   $1,0,e^{-\frac{\pi i}{9}},0$ \\ \hline
\multirow{2}{*}{$\mathrm{U}(1)_2$} & Trivial $1$ & $e^{\frac{2\pi i}{8}}, 0, e^{-\frac{2\pi i}{8}}, 1 $  & $1, 0, e^{\frac{13\pi i}{9}}, e^{\frac{13\pi i}{8}}$   \\ 
                        & Semion $s$  &  $e^{-\frac{2\pi i}{8}}, 1, e^{\frac{2\pi i}{8}}, 0$ &   $1,1, e^{\frac{\pi i}{9}}, 0$ \\ \hline
\end{tabular}
\caption{The phases of $\zeta_{n,a}$ and the partial rotation $\mathcal{T}_a(\frac{2\pi}{n})$ for $n=1,2,3,4$ in each topological sector of Ising and $\mathrm{U}(1)_2$. We write $0$ when the magnitude is vanishing.}
\label{tab:phases}
\end{table}

To compute partial rotation, we employ a cylinder geometry terminated with zigzag boundary condition on both ends as depicted in Fig.~\ref{fig:kitaev_main_text}, and we act on the left half of system with partial rotation.

The model is equivalent to a system of free Majorana fermions coupled to $\Z_2$ gauge field, by rewriting the qubits using Majorana fermion operators $c$, which act as dynamical free fermions, and $b$, which describes the $\mathbb{Z}_2$ gauge field.
As demonstrated in Supplemental Materials, the partial rotation for the state on the cylinder lying in the trivial sector can be expressed as
\begin{align}
        \mathcal{T}_1\left(\frac{2\pi}{n}\right)\propto\mathrm{Tr}\left(\frac{1+(-1)^F}{2}e^{-H_E}T_{\mathrm{A};\frac{2\pi}{n}}\right),
\end{align}
where $H_E$ is the entanglement Hamiltonian for the free fermion state in the A subsystem with the fixed flat $\Z_2$ gauge field, with the boundary condition in $y$ direction taken to be anti-periodic. The operator $(1+(-1)^F)/2$ gives a projector onto the Hilbert space with even fermion parity. Following~\cite{Qi2012momentumpolarization}, one can further evaluate it from the entanglement spectrum of the free Majorana fermions:
\begin{equation}
\begin{split}
    \mathcal{T}_1\left(\frac{2\pi}{n}\right) &\propto \prod_{m, k_y} \left[\frac{1+e^{ik_y L_y/n}}{2} + \frac{1-e^{ik_y L_y/n}}{2}\tanh\frac{\xi_{mk_y}}{2} \right] \\
    &+ \prod_{m, k_y} \left[\frac{1-e^{ik_y L_y/n}}{2} + \frac{1+e^{ik_y L_y/n}}{2}\tanh\frac{\xi_{mk_y}}{2} \right]
\end{split}
\label{eq:t1_from_majorana_spectrum}
\end{equation}
where $\xi_{mk_y}$ is the entanglement spectrum for $H_E$, carried by a quasiparticle with momentum $k_y$ in $y$ direction. 
Analogously, the partial rotation for the $\sigma$ sector is expressed in terms of the entanglement Hamiltonian $H_E^\sigma$ given by setting the periodic boundary condition in $y$ direction, $\mathcal{T}_{\sigma}(2\pi/n)\propto\mathrm{Tr}(e^{-H^\sigma_E}T_{\mathrm{A};2\pi/n})$, which can also be computed from entanglement spectrum of $H_E^\sigma$.

We show the result of this evaluation for $1, \sigma$ sectors in Fig.~\ref{fig:kitaev_main_text}. We see that $\mathrm{Arg}\left(\mathcal{T}_a\left({2\pi}/{n}\right)\right)$ converges to predicted values. 
We only present for $n \geq 3$ and $\left | \mathcal{T}_a\left({2\pi}/{n}\right) \right | > 0$. $\mathcal{T}_a(2\pi/n)$ is always real (no phase) for $n = 1$ and $2$ since the phase part exactly cancels. 

\begin{figure}[htbp]
    \centering
    \includegraphics[width = 0.48\textwidth]{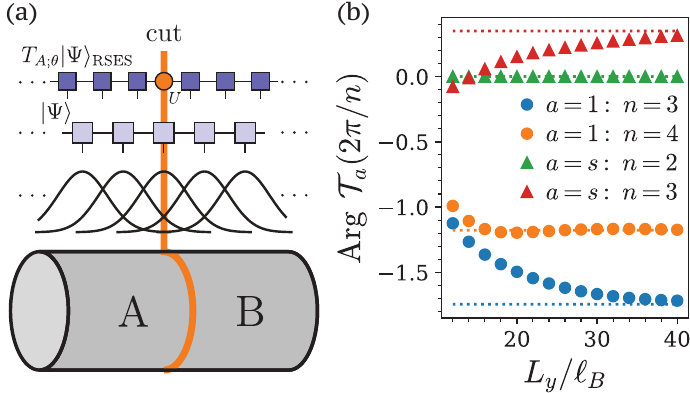}
    \caption{(a) A schematics of the infinite cylinder geometry and the LLL orbital basis of the MPS. Partial rotation along a real-space cut can be accomplished by acting a unitary operator on the auxiliary bond of the MPS obtained by the RSES algorithm. (b) $\operatorname{Arg} \mathcal{T}_a(2\pi/n)$ of the $\nu = 1/2$ bosonic Laughlin state extracted using Eq.~\eqref{eq:schmidt}.
    The dotted lines are the CFT predictions given in Table \ref{tab:phases}. }
    \label{fig:Laughlin}    
\end{figure}

The second example is the $\nu = 1/2$ bosonic Laughlin state, which realizes 
the
$\mathrm{U}(1)_2$ Chern-Simons theory. Its only non-trivial anyon is the semion $s$ with $\theta_s = i$.

The model we study is a half-filled lowest Landau level (LLL) of two-dimensional bosons with a contact interaction $V_0 = 1$ plus a small perturbation $\delta V_2 = 0.1$, where $V_m$ are the Haldane pseudopotentials \cite{Haldane_1994_pseudopotential,halperin_jain_cooper_2020}.
We consider an infinite cylinder geometry (Fig.~\ref{fig:Laughlin} (a)), and use infinite density matrix renormalization group (iDMRG) calculations~\cite{PhysRevB.91.045115} to obtain the infinite matrix product state (iMPS) representation of the ground state $|\Psi\rangle$.

Compared to other numerical methods, the MPS representation is advantageous for evaluating the action of partial rotation. If rotation is a good symmetry, the Schmidt states $|\alpha\rangle_{\mathrm{A/B}}$ across subsystems A and B have definite momentum $k_y^{\alpha}$ along the circumference. Thus, the action of partial rotation can be evaluated by
\begin{equation} \label{eq:schmidt}
\mathcal{T}_a(\theta) = \sum_{\alpha} \lambda_{\alpha}^2 e^{i k_y^{\alpha} L_y \theta},
\end{equation}
where $\lambda_{\alpha}$ is the corresponding Schmidt value. We can easily obtain both $k_y^{\alpha}$ and $\lambda_{\alpha}$ from the momentum label $\bar{K}_{\bar{n}_{\mathrm{B}};\alpha}$ and the Schmidt value $\lambda_{\bar{n}_{\mathrm{B}};\alpha}$ of the auxiliary bond $\bar{n}_{\mathrm{B}}$ across subsystems A and B.

For the $\nu = 1/2$ bosonic Laughlin state, we work in the Landau gauge and the corresponding LLL orbital basis.
To accelerate the calculation and obtain the momentum label mentioned above, we incorporate both particle number $\hat{C} =\sum_n \hat{C}_n \equiv \sum_n(\hat{N}_n-\nu)$ and momentum $\hat{K} =\sum_n \hat{K}_n \equiv \sum_n n(\hat{N}_n-\nu)$ conservation, where $\hat{N}_n$ is the number operator at site $n$. We find that $\mathcal{T}_a(2\pi/n)$ converges at bond dimension $\chi = 3200$, cylinder circumference $L_y = 40 \ell_B$ and onsite boson number cutoff $N_{\mathrm{boson}} = 5$.

 We note there are a few technical complications in applying Eq.~\eqref{eq:schmidt} to compute $\mathcal{T}_a(\theta)$, which we will scratch here and readers can find more details in Supplemental Materials. Firstly, there are a few ambiguities in extracting the physical momentum $k_y^{\alpha}$ from the momentum label $\bar{K}_{\bar{n}_{\mathrm{A}};\alpha}$. For iMPS, there is an overall ambiguity of momentum labels on auxiliary bonds. The magnetic translation symmetry in quantum Hall systems further tangles the momentum label $\bar{K}_{\bar{n} ; \alpha}= \langle \sum_{n < \bar{n}} \hat{K}_n\rangle_{\alpha}$ with the charge label $\bar{C}_{\bar{n} ; \alpha}= \langle \sum_{n < \bar{n}} \hat{C}_n \rangle_{\alpha}$ \cite{exactmps}.
 These ambiguities can be fixed by matching the entanglement spectrum and the edge CFT spectrum as elaborated in Supplemental Materials. 

Secondly, which topological sector subsystem A, B belongs to depends on the cut. The $\nu = 1/2$ bosonic Laughlin state has a two-fold ground state degeneracy, characterized by root configuration (pattern of zeros) $[01]$ and $[10]$ \cite{root_Haldane_2008,root_wen_2008}. It turns out that cutting through the LLL orbital center that corresponds to the $0$'s ($1$'s) bisects the system into two trivial sectors (semion) sectors. Finally, when we work in the LLL orbital basis, an auxiliary bond divides the system into two sets of LLL orbitals instead of two regions of physical space. This problem can be resolved using the real-space entanglement spectrum (RSES) algorithm developed in~\cite{exactmps}. We note that many of the technicalities discussed here are not specific to the $\nu = 1/2$ bosonic Laughlin state, but provide a general procedure for computing higher central charge of arbitrary wavefunction in the MPS form.

Finally we present the result of $\mathcal{T}_a(2\pi/n)$ in both the trivial and the semion sectors.
As shown in Fig.~\ref{fig:Laughlin} (b), $\mathcal{T}_a(2\pi/n)$ always converges to the expected phase as shown in Table \ref{tab:phases} at sufficiently large $L_y$.

\textit{Discussion --} In this Letter, we characterize the higher central charges 
$\{\zeta_n\}$ in terms of the partial rotation evaluated on a wavefunction of the (2+1)D bosonic topological order, and confirmed the prediction using the Kitaev honeycomb model and the $\nu=1/2$ bosonic Laughlin state.
Partial rotation can be implemented easily in quantum computing architectures with cheap SWAP gates, such as Rydberg atom arrays, which opens up another avenue to studying topological order directly on a quantum computer.
Together with topological entanglement entropy, partial rotation allows us to fully determine edge gappability of Abelian topological order.

It would be interesting to study applications of partial rotation to generic non-Abelian topological phases. Remarkably, even for non-Abelian phases, numerical results of $\{\TT_1(2\pi/n)\}$ put a tight constraint on the possible low-energy spectrum of the bulk-boundary system. For instance, 
suppose that we observed $\{\TT_1(2\pi/p_j)\}$ is a non-trivial phase for a set of distinct prime numbers $\{p_j\}$.
One can see that this leaves us two possibilities: 1. the edge is ungappable, or 2. the edge is gappable, where $\nfs$ must be divisible by $\prod_j p_j$. 
If the minimal $\nfs$ required for a gappable edge is large and physically unrealistic, one can essentially determine that the boundary must be ungappable. 

Notably, the lower bound $\nfs\ge \prod_j p_j$ for a gappable edge implies the lower bound for the number of anyons $r$ given by $r\ge r_0$, with $r_0$ the smallest integer satisfying $2^{2r_0/3+8}3^{2r_0/3} \ge \prod_j p_j$. This is derived from the fact that $\nfs$ of the bosonic topological order with $r$ distinct anyons has the upper bound $\nfs\le 2^{2r/3+8}3^{2r/3}$~\cite{Bruillard2015rankfiniteness}. It implies that the ground state on a torus must carry at least $r_0$-fold degeneracy in order to realize a gappable edge. 
This argument is reminiscent of the Lieb-Schultz-Mattis type theorems~\cite{Lieb1961, Oshikawa, Hastings}, which constrains the low-energy spectrum for a given input of the symmetry action on the ground state. 

Also, it would be interesting to extract the higher Hall conductivity proposed in~\cite{Kobayashi2022FQH}, which gives an obstruction to $\mathrm{U}(1)$ symmetry-preserving gapped boundary of the fermionic topological order with $\mathrm{U}(1)$ symmetry beyond electric Hall conductivity and $c_-$.

\textit{Acknowledgements --}  We thank Michael Zaletel, Roman Geiko, and Tianle Wang for helpful discussions.
RK is supported by the JQI postdoctoral fellowship at the University of Maryland.
TW is supported by the U.S.\ DOE, Office of Science, Office of Basic Energy Sciences, Materials Sciences and Engineering Division, under Contract No.\ DE-AC02-05-CH11231, within the van der Waals Heterostructures Program (KCWF16).
TS is supported by a fellowship from the Masason foundation.
RM is supported by the National Science Foundation under 
Award No.\ DMR-1848336.
SR is supported by the National Science Foundation under 
Award No.\ DMR-2001181, and by a Simons Investigator Grant from
the Simons Foundation (Award No.~566116).
This work is supported by
the Gordon and Betty Moore Foundation through Grant
GBMF8685 toward the Princeton theory program. This research used the Lawrencium computational cluster resource provided by the IT Division at the Lawrence Berkeley National Laboratory (Supported by the Director, Office of Science, Office of Basic Energy Sciences, of the U.S. Department of Energy under Contract No. DE-AC02-05CH11231).

\bibliography{bibliography.bib}

\onecolumngrid

\vspace{0.3cm}

\begin{center}
\Large{\bf Supplemental Materials}
\end{center}
\onecolumngrid


\section{Complete determination of edge gappability for bosonic Abelian phases}
Here, we show that the partial rotation $\{\mathcal{T}_1\left(2\pi/n\right)\}$, together with topological entanglement entropy~\cite{KitaevTEE, Levin2006Detecting}, can completely determine if the boundary of the (2+1)D bosonic Abelian topological phase is gappable or not. Concretely, a complete set of the gappability conditions can be obtained from a single wavefunction by the following steps:
\begin{enumerate}
    \item First, compute the topological entanglement entropy of the wavefunction to obtain the total quantum dimension $\mathcal{D}$ of the Abelian TQFT. Since the theory is Abelian, $\mathcal{D}^2$ gives the number of distinct anyons including the trivial one.
    \item Next, we want to extract the chiral central charge $c_-$. To do this, we notice that the Frobenius-Schur exponent $N_{\mathrm{FS}}$ of the Abelian TQFT must divide $2\mathcal{D}^2$ as explained in the next step. One can then extract $c_-$ by computing the partial rotation with $n = 2\mathcal{D}^2$,
    \begin{align}
        \mathcal{T}_1\left(\frac{2\pi}{2\mathcal{D}^2}\right) \propto e^{-2\pi i (\frac{2}{2\mathcal{D}^2}+2\mathcal{D}^2)\frac{c_-}{24}}~.   \end{align}
\item Finally, we want to extract the higher central charges to cover all obstructions to a gapped edge beyond $c_-$. For convenience, let us write the prime factorization of $\mathcal{D}^2$ as
\begin{align}
    \mathcal{D}^2 = p_1^{k_1}\times p_2^{k_2}\times \dots \times p_M^{k_M}
\end{align}
with $p_1< p_2<\dots < p_M$ prime numbers, and $k_1,k_2,\dots,k_M$ positive integers. Then, the Abelian TQFT $\mathcal{C}$ also factorizes as~\cite{kaidi2021higher}
\begin{align}
    \mathcal{C} = \mathcal{C}_{p_1} \boxtimes \mathcal{C}_{p_2} \boxtimes \dots \boxtimes  \mathcal{C}_{p_M}.
\end{align}
Here, $\boxtimes$ denotes the Deligne product of modular tensor categories that physically represents stacking of the theories. Each Abelian TQFT $\mathcal{C}_{p_j}$ contains the $p_j^{k_j}$ anyons including the trivial one. The Frobenius-Schur exponent also naturally admits the factorization as
\begin{align}
    N_{\mathrm{FS}} = N_1N_2\dots N_{M},
\end{align}
where $N_j$ gives the Frobenius-Schur exponent of the theory $\mathcal{C}_{p_j}$. Each $N_j$ is expressed by a positive integer power of $p_j$, $N_j = p_j^{l_j}$, where $l_j$ satisfies
\begin{align}
    \begin{cases}
        1\le l_1 \le k_1+1 & \text{if $p_1=2$,} \\
        1\le l_j \le k_j & \text{otherwise}.
    \end{cases}
\end{align}
This explains why $N_{\mathrm{FS}}$ must divide $2\mathcal{D}^2$ which is used in the previous step. 

Now we are ready to describe the all obstructions to a gapped edge in terms of the partial rotation. For this purpose, we recall that the Abelian TQFT admits a gapped edge if and only if all the theories $\mathcal{C}_{p_j}$ do~\cite{kaidi2021higher}. 
Then, each theory $\mathcal{C}_{p_j}$ admits a gapped edge if and only if the higher central charge $\zeta_n(\mathcal{C}_{p_j}) = 1$ for all $1\le n
\le N_j$ such that $\gcd(n,p_j) = 1$~\cite{drinfeld2010braided}. One can cover all the above higher central charges of $\mathcal{C}_{p_j}$ by computing the partial rotation $\mathcal{T}_1(2\pi/n)$ with
\begin{align}
    n = \frac{m\times 2\mathcal{D}^2}{p_j^{\tilde{k}_j}},
\end{align}
where $m$ scans all integers $1\le m\le p_j^{\tilde{k}_j}$ such that $\gcd(m,p_j)$ = 1. Here, we define the integers $\tilde{k}_j$ as
\begin{align}
    \begin{cases}
        \tilde{k}_1 = k_1+1 & \text{if $p_1 = 2$,} \\
        \tilde{k}_j=k_j & \text{otherwise}.
    \end{cases}
\end{align}
Let us check that above partial rotations exhaust all desired higher central charges of $\mathcal{C}_{p_j}$. Since $2\mathcal{D}^2/p_j^{\tilde{k}_j}$ is coprime with $p_j$ and divisible by $N_i$ for all $i$ such that $i\neq j$, one can write the partial rotation with $n = m\times 2\mathcal{D}^2/p_j^{\tilde{k}_j}$ as
\begin{align}
\begin{split}
    \mathcal{T}_1\left(\frac{2\pi}{n}\right) &\propto e^{-2\pi i (\frac{2}{n}+n)\frac{c_-}{24}}\times \sum_{a\in\mathcal{C}} \theta_a^n \\
    &= e^{-2\pi i (\frac{2}{n}+n)\frac{c_-}{24}}\times \left(\sum_{a\in\mathcal{C}_{p_1}} \theta_a^n \right) \times\dots \times \left(\sum_{a\in\mathcal{C}_{p_M}} \theta_a^n \right) \\
    &\propto e^{-2\pi i (\frac{2}{n}+n)\frac{c_-}{24}}\times\sum_{a\in\mathcal{C}_{p_j}} \theta_a^n \\
    &\propto e^{-2\pi i (\frac{2}{n}+n)\frac{c_-}{24}}\times\zeta_n(\mathcal{C}_{p_j}).
    \end{split}
\end{align}
One can then see that $\zeta_n(\mathcal{C}_{p_j})$ covers all the desired higher central charges by scanning $m$, since $n$ mod $N_j$ covers all $1\le n \le N_j$ such that $\gcd(n,p_j)=1$.
Since we do not know $N_{\mathrm{FS}}$ or $N_j$ in advance, we want to take the range of $m$ as $1\le m\le p_j^{\tilde{k}_j}$ so that $1\le m\le N_j$ is included.  This shows that the above $\mathcal{T}_1\left(2\pi/n\right)$ covers all gappability obstructions of $\mathcal{C}_{p_j}$.

Summarizing, we can obtain all gappability obstructions of the whole theory $\mathcal{C}$ beyond $c_-$, by computing the partial rotation 
$\mathcal{T}_1(2\pi/n)$ with
\begin{align}
    n = \frac{m\times 2\mathcal{D}^2}{p_j^{\tilde{k}_j}},
\end{align}
where $j$ scans all $1\le j\le M$, and $m$ scans all integers $1\le m\le p_j^{\tilde{k}_j}$ such that $\gcd(m,p_j)$ = 1. 
\end{enumerate}

\section{Partial rotation of a cylinder in a twisted sector}
Here we evaluate the partial rotation in the case where the state lies in the non-trivial topological sector labeled by an anyon $a$ of the topological order. Denoting the ground state as $\ket{\Psi_a}$, the partial rotation is again expressed as the Virasoro characters
\begin{align}
\begin{split}
    \mathcal{T}_a\left(\frac{2\pi}{n}\right) &= \frac{\mathrm{Tr}_a[e^{i{P}_l\frac{L}{n}}e^{-\frac{\xi_l}{v}H_l}]\mathrm{Tr}_a[e^{i{P}_r\frac{L}{n}}e^{-\frac{\xi_r}{v}H_r}]}{\mathrm{Tr}_a[e^{-\frac{\xi_l}{v}H_l}]\mathrm{Tr}_a[e^{-\frac{\xi_r}{v}H_r}]} \\
    &= \frac{\chi_a(\frac{i\xi_l}{L}+\frac{1}{n}) \chi_a(\frac{i\xi_r}{L}-\frac{1}{n}) }{\chi_a(\frac{i\xi_l}{L})\chi_a(\frac{i\xi_r}{L})}~,
    \label{eq:rotasCFTtwisted}
    \end{split}
\end{align}
where $\mathrm{Tr}_a$ denotes the trace in the twisted Hilbert space labeled by $a$, and $\chi_a(\tau)$ is the CFT character for the chiral primary $a$. The above quantity can be evaluated in the same logic as the main text; in our setup where $L \ll \xi_l$, the CFT characters for the left edge are approximated as
\begin{align}
\begin{split}
    \chi_a\left(\frac{i\xi_l}{L}\right) &\approx e^{-\frac{2\pi \xi_l}{L}(h_a-\frac{c_-}{24})}, \\ \chi_a\left(\frac{i\xi_l}{L}+\frac{1}{n}\right) &\approx e^{-\frac{2\pi \xi_l}{L}(h_a-\frac{c_-}{24})} e^{\frac{2\pi i}{n}(h_a-\frac{c_-}{24})}.
    \end{split}
    \label{eq:approxleftedge_twisted}
    \end{align}
Meanwhile, the characters for the right edge can be approximately computed by performing modular transformations as
\begin{align}
    \begin{split}
        \chi_a\left(\frac{i\xi_r}{L}\right) &= \sum_b S_{ab} \chi_b\left(\frac{iL}{\xi_r}\right) \approx \frac{d_a}{\mathcal{D}}e^{\frac{2\pi L}{\xi_r}\frac{c_-}{24}},
    \end{split}
\end{align}
\begin{align}
    \begin{split}
        \chi_a\left(\frac{i\xi_r}{L}-\frac{1}{n}\right) &= \sum_b (ST^n S)_{ab} \chi_b\left(\frac{iL}{n^2\xi_r}+\frac{1}{n}\right) \\
        &\approx (ST^n S)_{a,1}e^{-\frac{2\pi i}{n}\frac{c_-}{24}} e^{\frac{2\pi L}{n^2\xi_r}\frac{c_-}{24}} \\
        &= \frac{1}{\mathcal{D}} e^{-2\pi i(n + \frac{1}{n})\frac{c_-}{24}} e^{\frac{2\pi L}{n^2\xi_r}\frac{c_-}{24}}\sum_{b}S_{ab}d_b\theta_b^n,
        \end{split}
        \label{eq:approxrightcharacter_twisted}
\end{align}
where we assumed $n$ to be small, $n^2\ll L/\xi_r$, and used the approximation of the CFT characters for the right edge as
\begin{align}
\begin{split}
\chi_b\left(\frac{iL}{\xi_r}\right) &\approx e^{-\frac{2\pi L}{\xi_r}(h_b-\frac{c_-}{24})} \\
\chi_b\left(\frac{iL}{n^2\xi_r}+\frac{1}{n}\right) &\approx e^{\frac{2\pi i}{n}(h_b-\frac{c_-}{24})} e^{-\frac{2\pi L}{n^2\xi_r}(h_b-\frac{c_-}{24})} \\
    \end{split}
    \label{eq:characterapproxright}
\end{align}
Combining the above results, we obtain
\begin{align} 
\mathcal{T}_a\left(\frac{2\pi}{n}\right) \propto e^{\frac{2\pi i}{n} h_a-2\pi i (\frac{2}{n}+n)\frac{c_-}{24}} \times \sum_{b}S_{ab} d_b\theta_b^n~.
\label{eq:fullresult_twisted}
\end{align}


\section{Interpolating between higher central charge and momentum polarization}
In the main text, we have seen that the partial rotation $\mathcal{T}_1({2\pi}/{n})$ is proportional to higher central charge for a small positive integer $n$, while when $n\to\infty$ it rather gives the momentum polarization $\lambda$ in Eq.~\eqref{eq:momentumpolarization}. Here, we describe how these two different behaviors are interpolated by changing the value of $n$ from a small integer to infinity. 
Following Eq.~\eqref{eq:rotasCFTtwisted} and Eq.~\eqref{eq:approxleftedge_twisted}, we express the partial rotation as
\begin{align}
\begin{split}
    \mathcal{T}_a\left(\frac{2\pi}{n}\right) &=e^{\frac{2\pi i}{n}(h_a-\frac{c_-}{24})}\times\frac{\chi_a(\frac{i\xi_r}{L}-\frac{1}{n}) }{\chi_a(\frac{i\xi_r}{L})}. 
    \end{split}
\end{align}
The CFT characters on the right edge is approximated depending on the scale of $n$ compared with $L/\xi_r$. Below we compute it in cases of the value of $n$.
\begin{itemize}
    \item $n^2\ll L/\xi_r$. The approximation in Eq.~\eqref{eq:approxrightcharacter_twisted} is valid for such $n$ and we have 
    \begin{align} 
\mathcal{T}_a\left(\frac{2\pi}{n}\right) \propto e^{\frac{2\pi i}{n} h_a-2\pi i (\frac{2}{n}+n)\frac{c_-}{24}} \times \sum_{b}S_{ab} d_b\theta_b^n~.
\end{align}
\item $L/\xi_r\ll n^2$.
When $n$ satisfies $L/\xi_r \ll n^2$, Eq.~\eqref{eq:approxrightcharacter_twisted} is no longer valid, where the alternative approximation is available
\begin{align}
    \begin{split}
        \chi_a\left(\frac{i\xi_r}{L}-\frac{1}{n}\right) &= \sum_b S_{ab} \chi_b\left(-\frac{1}{\frac{i\xi_r}{L}-\frac{1}{n}}\right) \\
        &\approx \sum_b S_{ab}\exp\left({-2\pi i \left(h_b-\frac{c_-}{24}\right)\frac{1}{\frac{i\xi_r}{L}-\frac{1}{n}}}\right) \\
        & \approx \frac{d_a}{\mathcal{D}}\exp\left({2\pi i \frac{c_-}{24}\frac{1}{\frac{i\xi_r}{L}-\frac{1}{n}}}\right)~.       \end{split}
\label{eq:approxrightcharacter_twisted_largen}
\end{align}
The partial rotation is then given by
\begin{align}
\begin{split}
    \mathcal{T}_a\left(\frac{2\pi}{n}\right) =& e^{\frac{2\pi i}{n}(h_a-\frac{c_-}{24})}\exp\left({2\pi i \frac{c_-}{24}\frac{1}{\frac{i\xi_r}{L}-\frac{1}{n}}}-\frac{2\pi L}{\xi_r}\frac{c_-}{24}\right) \\
    =& \exp\left(\frac{2\pi i}{n}(h_a-\frac{c_-}{24})-2\pi i\frac{c_-}{24} \frac{\frac{1}{n}}{\frac{\xi_r^2}{L^2}+ \frac{1}{n^2}} \right) \\
    &\times \exp\left(2\pi \frac{c_-}{24} \frac{\frac{\xi_r}{L}}{\frac{\xi_r^2}{L^2}+\frac{1}{n^2}}-2\pi\frac{c_-}{24}\frac{L}{\xi_r} \right)~. \end{split}
\end{align}
By taking the limit $n\to \infty$, it reproduces the expression of momentum polarization in Eq.~\eqref{eq:momentumpolarization}.

\end{itemize}

\section{Relationship between Entanglement spectrum and edge theory in non-chiral topological phases}
Here, we describe the correspondence between entanglement spectrum and physical edge CFT in the (2+1)D bosonic topological phases, following the cut-and-glue approach introduced in Ref.~\cite{Qi2012entanglement}. While the correspondence is discussed for chiral topological phases whose edge CFT is holomorphic in Ref.~\cite{Qi2012entanglement}, we slightly generalize their argument to the case of possibly non-chiral topological phases, verifying the cut-and-glue approach for generic non-chiral phases which are of our interest.

We consider a (possibly non-chiral) bosonic topological ordered state in (2+1)D. Let us take bipartition of the state into subsystems A and B.
The physical Hamiltonian can be described as $H_\mathrm{A} + H_{\mathrm{B}} + \lambda H_{\mathrm{AB}}$, where $H_{\mathrm{AB}}$ denotes the term that couples the subsystem on A and B subregions. 
We want to study the reduced density matrix $\rho_{\mathrm{A}}$ for the A subsystem, where the entanglement Hamiltonian is given by $\rho_{\mathrm{A}} = e^{-H_E}$. We will see that $\rho_{\mathrm{A}}$ can be described by the thermal density matrix of the edge CFT at high temperature.

If the coupling $\lambda$ is taken much smaller than the bulk energy gap, the main effect of the term $H_{\mathrm{AB}}$ is to induce a perturbation to the edge theories on the cut. Let us assume that the entanglement property of the system can be well-described by taking the small coupling $\lambda$, which is a natural assumption when $H_{\mathrm{AB}}$ induces a relevant perturbation to CFT that can gap out the edge spectrum with small coupling.
In that case, the entanglement spectrum can be effectively described by the edge theories on the cut,
\begin{align}
    H= H_{\mathrm{A}}^{\mathrm{CFT}} + H_{\mathrm{B}}^{\mathrm{CFT}} + H_{\mathrm{int}},
\end{align}
where $H_{\mathrm{A}}^{\mathrm{CFT}}, H_{\mathrm{B}}^{\mathrm{CFT}}$ describe the edge CFT at the cut, and $H_{\mathrm{int}}$ introduces the inter-edge perturbation. We consider the case where $H_{\mathrm{A}}^{\mathrm{CFT}}, H_{\mathrm{B}}^{\mathrm{CFT}}$ are not necessarily chiral.
To study the ground state of this perturbed Hamiltonian $H$, we recall that the relevant perturbation $H_{\mathrm{int}}$ corresponds to a conformally invariant boundary condition of the CFT $H_{\mathrm{A}}^{\mathrm{CFT}} + H_{\mathrm{B}}^{\mathrm{CFT}}$, where the boundary condition can be explicitly obtained by turning on the perturbation on the half region of the 1d space. This perspective allows us to describe the ground state of the perturbed Hamiltonian $H$ using the toolkit of boundary conformal field theory (BCFT).

To do this, let us consider a ``quenching'' process of the system, whose dynamics is given by the Hamiltonian $H$ at the time $t<0$, and we suddenly turn off the perturbation as $\lambda=0$ at $t=0$. In the (1+1)D spacetime, this quench is regarded as having the boundary condition of the unperturbed CFT at $t=0$, so the initial state at $t=0$ can be effectively described by the boundary state for the boundary condition that corresponds to $H_{\mathrm{int}}$. 

Let us now describe this boundary state. Suppose that the state of $H_{\mathrm{A}}^{\mathrm{CFT}}$ lies in the topological sector of the CFT labeled by the anyon $a$ in the bulk.  Since we want to work on a non-chiral topological phase in general, we assume that the anyon $a$ in the bulk corresponds to a generic non-chiral primary field $(\alpha,\overline{\beta})$ on the edge CFT, where $\alpha$ (resp.~$\overline{\beta}$) labels chiral (resp.~anti-chiral) primary. Accordingly, the state of $H_{\mathrm{B}}^{\mathrm{CFT}}$ lies in the topological sector $(\beta,\overline{\alpha})$. The state of the combined theory $H_{\mathrm{A}}^{\mathrm{CFT}} + H_{\mathrm{B}}^{\mathrm{CFT}}$ is defined in the Hilbert space $(\mathcal{H}_{\alpha}\otimes\mathcal{H}_{\beta})\otimes(\mathcal{H}_{\overline{\alpha}}\otimes\mathcal{H}_{\overline{\beta}})$, where $\mathcal{H}_{\alpha}$ denotes the holomorphic CFT Hilbert space labeled by the chiral primary $\alpha$. Note that $\mathcal{H}_{\alpha}\otimes \mathcal{H}_{\overline{\beta}}$ (resp.~$\mathcal{H}_{\beta}\otimes \mathcal{H}_{\overline{\alpha}}$) corresponds to the Hilbert space supported on the A (resp.~B) subsystem. In this Hilbert space, the boundary state of the combined theory $H_{\mathrm{A}}^{\mathrm{CFT}} + H_{\mathrm{B}}^{\mathrm{CFT}}$ can be described by the Ishibashi state
\begin{align}
    \ket{a_{\mathrm{bdry}}} = \sum_N \ket{(\alpha,\beta),N}\otimes \ket{\overline{(\alpha,\beta)},N}
\end{align}
where $\ket{(\alpha,\beta),N}$ denotes the state in the holomorphic Hilbert space $\mathcal{H}_{\alpha}\otimes\mathcal{H}_{\beta}$ with the label $N$, and the sum of $N$ runs all the orthonormal basis of $\mathcal{H}_{\alpha}\otimes\mathcal{H}_{\beta}$. 
Since each basis state can be written as $\ket{V_\alpha}\otimes\ket{V_\beta}$, the sum is rewritten as
\begin{align}
    \ket{a_{\mathrm{bdry}}} = \sum_{N_\alpha, N_\beta} \ket{\alpha,N_\alpha}\ket{\beta,N_{\beta}} \otimes \ket{\overline{\alpha},N_\alpha}\ket{\overline{\beta},N_{\beta}}    \end{align}
using the label $N_\alpha, N_\beta$ for the orthonormal basis of $\mathcal{H}_\alpha, \mathcal{H}_\beta$ respectively.
To obtain the reduced density matrix $\rho_\mathrm{A}$, we trace out the Hilbert space $\mathcal{H}_{\beta}\otimes \mathcal{H}_{\overline{\alpha}}$ for the B subsystem. We then have
\begin{align}
    \rho_{\mathrm{A}} = \sum_{N_\alpha, N_\beta} [\ket{\alpha,N_\alpha}\ket{\overline{\beta},N_\beta}] [\bra{\alpha,N_\alpha}\bra{\overline{\beta},N_\beta}]
\end{align}
which is regarded as the thermal density matrix of CFT for the A subsystem at infinite temperature. This demonstrates that the entanglement spectrum is described by that of the physical edge CFT at high temperature, including the non-chiral topological phases.

\section{Preparing the trivial sector of Kitaev honeycomb model on a cylinder}

To compute higher central charge of the topologically ordered state, one needs to prepare a low-energy state on the cylinder that lies in the trivial sector labeled by the trivial anyon $1$. In this section, we describe how to obtain such a state for the non-Abelian gapped phase of the Kitaev honeycomb model on a cylinder.

The Kitaev honeycomb model is described in terms of a free fermion model coupled to a dynamical $\Z_2$ gauge field. The $\Z_2$ vortex of the gauge field is regarded as a single quasiparticle excitation $\sigma$, and the low-energy state is realized by flat $\Z_2$ gauge field without a vortex on each plaquette.
The distinct low-energy states on a cylinder then correspond to the different configurations of flat $\Z_2$ gauge field labeled by the boundary conditions of the free fermion model. In particular, the anti-periodic boundary condition (APBC) in the $y$ direction (i.e., circumference) of the cylinder corresponds to 1 or $\psi$ sector, while the periodic (PBC) corresponds to the $\sigma$ sector. To further obtain the trivial sector $1$ instead of $\psi$, we need to carefully study the holonomy of $\Z_2$ gauge field in $x$ direction.
As we will see below, fixing the holonomy of the $\Z_2$ gauge field in $x$ direction gives the state labeled by a non-simple anyon $1+\psi$. 
To obtain a state in the trivial sector $1$, it turns out that we need to take a superposition over the two distinct states that correspond to different $\Z_2$ holonomy in $x$ direction.

Let us move to the explicit description of the states of the Kitaev honeycomb model.
We assume that the system size $L_y$ in the $y$ direction is even.
Recall that to solve this model, we introduce four Majorana fermion operators $b^x_j, b^y_j, b^z_j, c_j$ on each site $j$, and the model then becomes equivalent to the free fermion model for the Majorana fermions $c$, coupled to the dynamical $\Z_2$ gauge field $u_{ij}:=ib_ib_j$~\cite{Kitaevanyons}.
Let us fix one configuration of flat $\Z_2$ gauge field of the Kitaev honeycomb model $\{u_{ij}\}$ that corresponds to the APBC in the $y$ direction.
For this given configuration of the $\Z_2$ gauge field $\{u_{ij}\}$, the ground state of the Kitaev honeycomb model is given in the form of
\begin{align}
    \ket{\Psi} = P\left[\ket{\Psi_F(\{u_{ij}\})}\otimes \ket{\{u_{ij}\}}\right],
\end{align}
where the state is given by the tensor product of the state of the free fermion 
$\ket{\Psi_F(\{u_{ij}\})}$ and the state $\ket{\{u_{ij}\}}$ for the dynamical $\Z_2$ field. $P:= \prod_j \frac{1+D_j}{2}$ denotes a projection onto the gauge invariant Hilbert space, where $D_j:= ib^x_jb^y_jb^z_jc$ generates the gauge transformation at the vertex $j$.

The distinct state with different $\Z_2$ holonomy in the $x$ direction can be obtained by considering a straight line $l$ extended in the $y$ direction cutting the links, and then shifting the $\Z_2$ gauge field $\{u_{ij}\}$ along the links cut by the line $l$. Let us denote the shifted $\Z_2$ gauge field as $\{u'_{ij}\}$, and the corresponding state of the Kitaev honeycomb model as
\begin{align}
    \ket{\Psi'} = P\left[\ket{\Psi_F(\{u'_{ij}\})}\otimes \ket{\{u'_{ij}\}}\right].
\end{align}
We aim to show that the two states $\ket{\Psi} \pm \ket{\Psi'}$ correspond to the states in the sector 1 and $\psi$, though we do not try to specify which of $\ket{\Psi} \pm \ket{\Psi'}$ corresponds to the trivial sector. To see this, we evaluate the partial rotation for the states $\ket{\Psi} \pm \ket{\Psi'}$, since the partial rotation can diagnose the distinct topological sector.

To perform the partial rotation, we need to take the bipartition of the cylinder into A and B subsystems. For convenience, let us perform the bipartition along the line $l$ that we used for shifting the $\Z_2$ gauge field from $\{u_{ij}\}$ to $\{u'_{ij}\}$. 
Without loss of generality, we can take the gauge where $u_{ij} = 1$ on each link cut by the line $l$, while $u'_{ij}=-1$ on these links.

Since the link variables $u_{ij}$, $u'_{ij}$ cut by the line $l$ overlap
the Hilbert space of A and B subsystems, it is convenient to redefine link variables near the cut so that the newly defined ones lie solely in the A or B subsystem. This can be done by following the argument in Appendix A of Ref.~\cite{Qi2012momentumpolarization}. We recall that the link variable is given by the doublet of Majorana fermions as $u_{ij} = ib_{i}b_j$, and let us write the Majorana fermions at the plaquette on the cut as $b_1,b_2,b_3,b_4$ as shown in Fig.~\ref{fig:wij}. The newly defined link variables are then given by $w_{13}:= ib_1b_3, w_{42} = ib_4b_2$. The local Hilbert space spanned by these four Majorana fermions is identified as the Hilbert space of two qubits, where we have the identification of operators as
\begin{align}
    u_{12} = Z_1,\quad u_{34} = Z_2, \quad w_{13} = X_1X_2, \quad w_{42} = -Y_1Y_2. 
\end{align}
Some of the eigenstates of $w_{13}, w_{42}$ are expressed as
\begin{align}
    \ket{++} = \frac{1}{\sqrt{2}} (\ket{00} + \ket{11}), \quad \ket{--} = \frac{1}{\sqrt{2}} (\ket{00} - \ket{11})
    \end{align}
where $+, -$ denotes eigenvalues of $w_{13}, w_{42}$, while $0,1$ denotes those of $u_{12},u_{34}$. 

\begin{figure}[htbp]
    \centering
    \includegraphics[width = 0.23\textwidth]{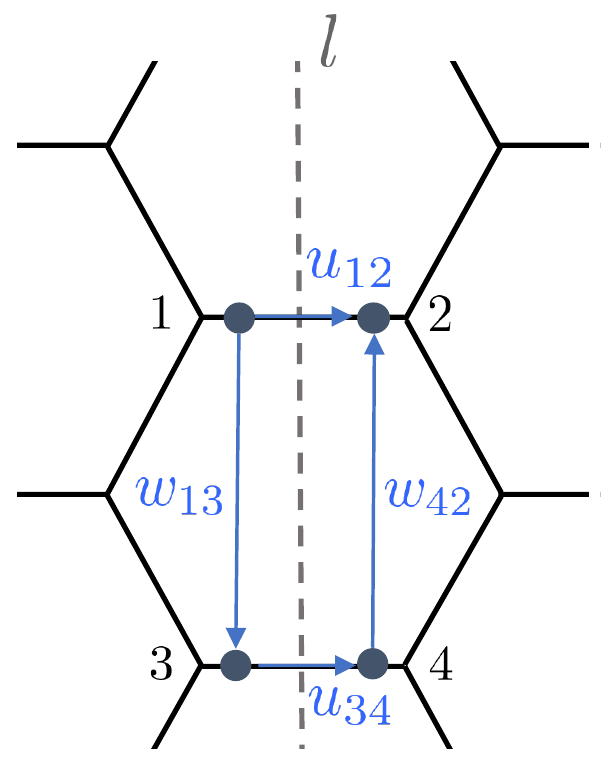}
    \caption{The link variables near the bipartition. The thick dots denote the Majorana fermions $b_j$ for $1\le j\le 4$. The link variables $u_{12}=ib_1b_2,u_{34}=ib_3b_4$ are cut by the line $l$, and we define the new ones $w_{13}=ib_1b_3, w_{42}=ib_4b_2$. }
    \label{fig:wij}    
\end{figure}

Noting that $u_{ij}=1$ on the edges cut by the bipartition, the state for the $\Z_2$ gauge field $\ket{\{u_{ij}\}}$ is decomposed into the Hilbert spaces of A, B subsystems as
\begin{align}
    \ket{\{u_{ij}\}} = 2^{-\frac{L_y}{4}}\sum_{w_{ij} = \pm 1} \ket{\{u^{\mathrm{A}}_{ij}, w_{ij}\}}_{\mathrm{A}}\otimes \ket{\{u^{\mathrm{B}}_{ij}, w_{ij}\}}_{\mathrm{B}}~.
\end{align}
Meanwhile, since $u'_{ij}=-1$ on the edges cut by the bipartition, the state for the shifted $\Z_2$ gauge field $\ket{\{u'_{ij}\}}$ is decomposed as
\begin{align}
    \ket{\{u'_{ij}\}} = 2^{-\frac{L_y}{4}}\sum_{w_{ij} = \pm 1} (-1)^{|w|} \ket{\{u^{\mathrm{A}}_{ij}, w_{ij}\}}_{\mathrm{A}}\otimes \ket{\{u^{\mathrm{B}}_{ij}, w_{ij}\}}_{\mathrm{B}}~,
\end{align}
where $|w|$ is the mod 2 number of links on the A subsystem with $w_{ij}=-1$. Note that $u_{ij}=u'_{ij}$ away from the cut, so each Schmidt state in the A, B subsystems are the same in the expression of $\ket{\{u_{ij}\}}$ and $\ket{\{u'_{ij}\}}$.

Let us also write the Schmidt decomposition of the free fermion state $\ket{\Psi_F(\{u_{ij}\})}$ as
\begin{align}
    \ket{\Psi_F(\{u_{ij}\})} = \sum_N \alpha_N \ket{\Psi^{\mathrm{A}}_{N}(\{u^{\mathrm{A}}_{ij}\})} \otimes \ket{\Psi^{\mathrm{B}}_{N}(\{u^{\mathrm{B}}_{ij}\})}~.
\end{align}
Next, we assume that the low-energy state for the free fermion $\ket{\Psi_F(\{u'_{ij}\})}$ under the shifted $\Z_2$ gauge field $\{u'_{ij}\}$ can be chosen to be identical to that under the trivial $\Z_2$ gauge field $\{u_{ij}\}$,
\begin{align}
    \ket{\Psi_F(\{u'_{ij}\})} = \ket{\Psi_F(\{u_{ij}\})},
\end{align}
and hence has the same Schmidt decomposition. 
Below, we show that $\ket{\Psi}\pm\ket{\Psi'}$ with the above $\ket{\Psi'}$ corresponds to the 1 or $\psi$ sector.

The reduced density matrix of $\ket{\Psi}$ for the A subsystem is given by~\cite{Qi2012momentumpolarization}
\begin{align}
    \rho_{\mathrm{A}} := \mathrm{Tr}_{\mathrm{B}}(\ket{\Psi}\bra{\Psi}) \propto \sum_{w_{ij} = \pm 1} \sum_N |\alpha_N|^2 P_{\mathrm{A}}\left[\ket{\Psi_N^{\mathrm{A}}(\{u^{\mathrm{A}}_{ij}\})}\ket{\{u^{\mathrm{A}}_{ij},w_{ij}\}}\right] P_{\mathrm{A}}\left[\bra{\Psi_N^{\mathrm{A}}(\{u^{\mathrm{A}}_{ij}\})}\bra{\{u^{\mathrm{A}}_{ij},w_{ij}\}}\right],
\end{align}
where $P_{\mathrm{A}} = \prod_{j\in\mathrm{A}}\frac{1+D_j}{2}$ denotes the projector supported on the A subsystem. Noting that the product of gauge transformations in the A subsystem $\prod_{j\in\mathrm{A}} D_j$ leaves the link variables $\{u^{\mathrm{A}}_{ij}, w_{ij}\}$ invariant, one can rewrite the expression of the state as
\begin{align}
    P_{\mathrm{A}}\left[\ket{\Psi_N^{\mathrm{A}}(\{u^{\mathrm{A}}_{ij}\})}\ket{\{u^{\mathrm{A}}_{ij},w_{ij}\}}\right] \propto \sum_{\{u'^{\mathrm{A}}_{ij}, w'_{ij}\}\approx\{u^{\mathrm{A}}_{ij}, w_{ij}\}} \frac{1+\prod_{j\in\mathrm{A}}D_j}{2} \ket{\Psi_N^{\mathrm{A}}(\{u'^{\mathrm{A}}_{ij}\})}\ket{\{u'^{\mathrm{A}}_{ij},w'_{ij}\}}
\end{align}
where the sum is taken over gauge fields $\{u'^{\mathrm{A}}_{ij},w'_{ij}\}$ that are equivalent to $\{u^{\mathrm{A}}_{ij},w_{ij}\}$ up to gauge transformations. The reduced density matrix is then expressed as
\begin{align}
    \rho_{\mathrm{A},\ket{\Psi}}\propto \sum_{w_{ij} = \pm 1}\sum_{\substack{\{u'^{\mathrm{A}}_{ij}, w'_{ij}\}\approx\{u^{\mathrm{A}}_{ij},w_{ij}\} \\\{u''^{\mathrm{A}}_{ij}, w''_{ij}\}\approx\{u^{\mathrm{A}}_{ij},w_{ij}\} }  } \sum_N |\alpha_N|^2\frac{1+\prod_{j\in\mathrm{A}}D_j}{2} \ket{\Psi_N^{\mathrm{A}}(\{u'^{\mathrm{A}}_{ij}\})}\ket{\{u'^{\mathrm{A}}_{ij},w'_{ij}\}}\bra{\Psi_N^{\mathrm{A}}(\{u''^{\mathrm{A}}_{ij}\})}\bra{\{u''^{\mathrm{A}}_{ij},w''_{ij}\}}
    \frac{1+\prod_{j\in\mathrm{A}}D_j}{2} 
    \end{align}
Also, the reduced density matrix for the state $\ket{\Psi}\pm\ket{\Psi'}$ is given by
\begin{align}
\begin{split}
    \rho_{\mathrm{A},\ket{\Psi}\pm\ket{\Psi'}}\propto &\sum_{w_{ij} = \pm 1} \frac{1\pm(-1)^{|w|}}{2}\sum_{\substack{\{u'^{\mathrm{A}}_{ij}, w'_{ij}\}\approx\{u^{\mathrm{A}}_{ij},w_{ij}\} \\\{u''^{\mathrm{A}}_{ij}, w''_{ij}\}\approx\{u^{\mathrm{A}}_{ij},w_{ij}\} }  } \sum_N\\
    &|\alpha_N|^2\frac{1+\prod_{j\in\mathrm{A}}D_j}{2} \ket{\Psi_N^{\mathrm{A}}(\{u'^{\mathrm{A}}_{ij}\})}\ket{\{u'^{\mathrm{A}}_{ij},w'_{ij}\}}\bra{\Psi_N^{\mathrm{A}}(\{u''^{\mathrm{A}}_{ij}\})}\bra{\{u''^{\mathrm{A}}_{ij},w''_{ij}\}}
    \frac{1+\prod_{j\in\mathrm{A}}D_j}{2}.
    \end{split}
    \end{align}
The partial rotation of these reduced density matrices for $\ket{\Psi},\ket{\Psi}\pm\ket{\Psi'}$ are evaluated as
\begin{align}
    \begin{split}
        \mathcal{T}_{\ket{\Psi}}\left(\frac{2\pi}{n}\right) &\propto\sum_{w_{ij}=\pm 1}\sum_{N}|\alpha_N|^2\bra{\Psi_N^{\mathrm{A}}(\{T_{\frac{L_y}{n}}u^{\mathrm{A}}_{ij}\})}\frac{1+\prod_{j\in\mathrm{A}}D_j}{2} T_{\mathrm{A};\frac{L_y}{n}}\ket{\Psi_N^{\mathrm{A}}(\{u^{\mathrm{A}}_{ij}\})}~, \\
        \mathcal{T}_{\ket{\Psi}\pm\ket{\Psi'}}\left(\frac{2\pi}{n}\right) &\propto \sum_{w_{ij}=\pm 1}\sum_{N}|\alpha_N|^2\bra{\Psi_N^{\mathrm{A}}(\{T_{\frac{L_y}{n}}u^{\mathrm{A}}_{ij}\})}\frac{1\pm(-1)^{|w|}}{2}\frac{1+\prod_{j\in\mathrm{A}}D_j}{2} T_{\mathrm{A};\frac{L_y}{n}}\ket{\Psi_N^{\mathrm{A}}(\{u^{\mathrm{A}}_{ij}\})}~,
    \end{split}
\end{align}
where we assume $L_y$ is the integer multiple of $n$, and $T_{\mathrm{A};\frac{L_y}{n}}$ is the translation along the $y$ direction by $L_y/n$ lattice units. $\{T_{\frac{L_y}{n}}u^{\mathrm{A}}_{ij}\}$ denotes the gauge field obtained by acting the translation $T_{\mathrm{A};\frac{L_y}{n}}$ on the configuration $\{u_{ij}\}$. The above expressions for the partial rotation can further be simplified by writing the operator $\prod_{j\in\mathrm{A}}D_j$ as
\begin{align}
    \prod_{j\in\mathrm{A}}D_j = (-1)^F(-1)^{|w|}\prod u^{\mathrm{A}}_{ij},
\end{align}
where $(-1)^F$ is the fermion parity operator of the A subsystem acting only on the free fermion state. 
The product over link variables are dependent on the orientations of the links, but its detailed definition is not discussed here. Since the projector $\frac{1+\prod_{j\in\mathrm{A}}D_j}{2}$ enforces the relation $(-1)^F = (-1)^{|w|} \prod u^{\mathrm{A}}_{ij}$ and we sum over all possible choices of $w_{ij}=\pm 1$, the half of the projectors in the summand of $w_{ij}$ becomes $\frac{1+(-1)^F}{2}$, while the other half becomes $\frac{1-(-1)^F}{2}$. Hence we have
\begin{align}
    \sum_{w_{ij}=\pm 1}\frac{1+\prod_{j\in\mathrm{A}}D_j}{2}\propto \frac{1+(-1)^F}{2} + \frac{1-(-1)^F}{2} = 1.
    \end{align}
Meanwhile, for the operator $\frac{1\pm(-1)^{|w|}}{2}\frac{1+\prod_{j\in\mathrm{A}}D_j}{2}$, the contribution of either $\frac{1+(-1)^F}{2}$ or $\frac{1+(-1)^F}{2}$ is suppressed by the presence of the $\frac{1\pm(-1)^{|w|}}{2}$ factor.
We hence have
\begin{align}
    \sum_{w_{ij}=\pm 1}\frac{1\pm(-1)^{|w|}}{2}\frac{1+\prod_{j\in\mathrm{A}}D_j}{2}\propto\quad  \frac{1+(-1)^F}{2}  \quad \text{or} \quad \frac{1-(-1)^F}{2}.
    \end{align}
Eventually, one can express the partial rotation of each state in the form only involving the free fermion state as
\begin{align}
    \begin{split}
        \mathcal{T}_{\ket{\Psi}}\left(\frac{2\pi}{n}\right) &\propto\sum_{N}|\alpha_N|^2\bra{\Psi_N^{\mathrm{A}}(\{T_{\frac{L_y}{n}}u^{\mathrm{A}}_{ij}\})} T_{\mathrm{A};\frac{L_y}{n}}\ket{\Psi_N^{\mathrm{A}}(\{u^{\mathrm{A}}_{ij}\})}~,
  \end{split}  
  \end{align}

  \begin{align}
           \mathcal{T}_{\ket{\Psi}\pm\ket{\Psi'}}\left(\frac{2\pi}{n}\right) \propto 
           \begin{cases}
           &\sum_{N}|\alpha_N|^2\bra{\Psi_N^{\mathrm{A}}(\{T_{\frac{L_y}{n}}u^{\mathrm{A}}_{ij}\})}\frac{1+(-1)^F}{2}T_{\mathrm{A};\frac{L_y}{n}}\ket{\Psi_N^{\mathrm{A}}(\{u^{\mathrm{A}}_{ij}\})}~, \\
           &\quad \text{or} \\
           & \sum_{N}|\alpha_N|^2\bra{\Psi_N^{\mathrm{A}}(\{T_{\frac{L_y}{n}}u^{\mathrm{A}}_{ij}\})}\frac{1-(-1)^F}{2}T_{\mathrm{A};\frac{L_y}{n}}\ket{\Psi_N^{\mathrm{A}}(\{u^{\mathrm{A}}_{ij}\})}~.
           \end{cases}
  \end{align}
  Let us write the entanglement Hamiltonian of the free fermion state $\ket{\Psi_F(\{u^{\mathrm{A}}_{ij}\})}$ as $H_E$. The above partial rotations are then simply expressed as
  \begin{align}
      \mathcal{T}_{\ket{\Psi}}\left(\frac{2\pi}{n}\right) \propto \mathrm{Tr}\left(e^{-H_E}T_{\mathrm{A};\frac{L_y}{n}}\right),
\end{align}
\begin{align}
    \mathcal{T}_{\ket{\Psi}\pm\ket{\Psi'}}\left(\frac{2\pi}{n}\right) \propto \quad \mathrm{Tr}\left(\frac{1+(-1)^F}{2}e^{-H_E}T_{\mathrm{A};\frac{L_y}{n}}\right)\quad \text{or} \quad \mathrm{Tr}\left(\frac{1-(-1)^F}{2}e^{-H_E}T_{\mathrm{A};\frac{L_y}{n}}\right)~.
        \end{align}
At the CFT level, the projector $(1+(-1)^F)/2$ (resp.~$(1-(-1)^F)/2$) corresponds to the projector onto the fermion parity even (resp.~odd) Hilbert space, which corresponds to the sector 1 (resp.~$\psi$). We hence have
\begin{align}
        \mathcal{T}_1\left(\frac{2\pi}{n}\right)\propto\mathrm{Tr}\left(\frac{1+(-1)^F}{2}e^{-H_E}T_{\mathrm{A};\frac{L_y}{n}}\right), \quad \mathcal{T}_\psi\left(\frac{2\pi}{n}\right)\propto\mathrm{Tr}\left(\frac{1-(-1)^F}{2}e^{-H_E}T_{\mathrm{A};\frac{L_y}{n}}\right). 
\end{align}
In the main text, we directly confirm the first relation for $\mathcal{T}_1\left(\frac{2\pi}{n}\right)$ by comparing the CFT results for LHS and the numerical results for RHS.


\section{Extracting higher central charge of the Kitaev honeycomb model}

\begin{figure}[htbp]
    \centering
     \includegraphics[width = 0.48\textwidth]{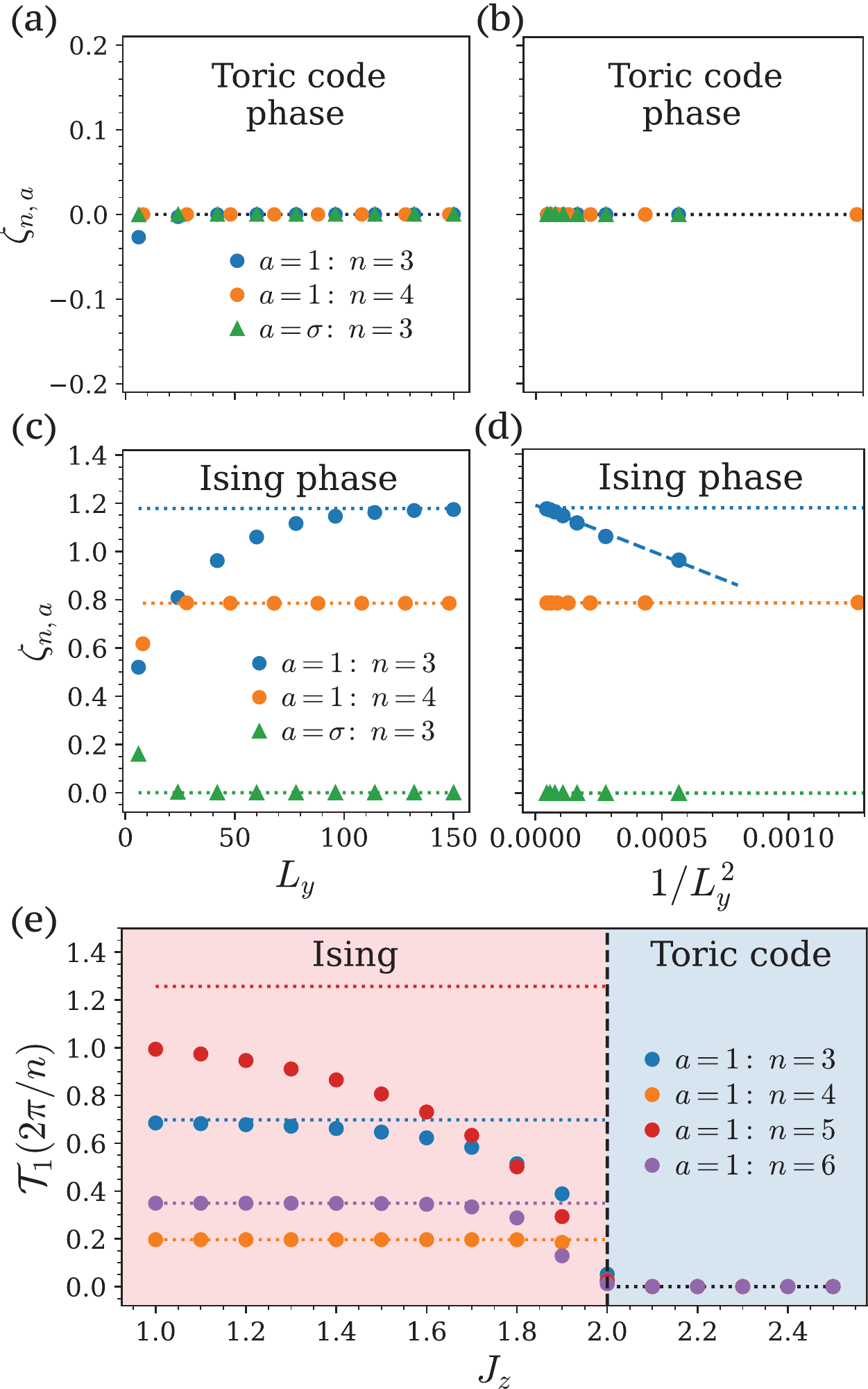}
    \caption{(a) Higher central charge $\zeta_{n}$ of the toric code phase of the Kitaev model, extracted from $\mathcal{T}_a(2\pi/n)$ in Fig.~\ref{fig:kitaev_main_text} (b) and chiral central charge $c_- = 0$ using the main result Eq.~\eqref{eq:mainresult}, computed using $J_x = 2.5, J_y = J_z = 1, \kappa = 0$. The dotted lines are the CFT predictions as usual.
    (b) Extrapolating $\zeta_{n}$ as a function of $1/L_y^2$
    (c) Twisted higher central charge $\zeta_{n,a}$ of the Ising phase of the Kitaev model, extracted using chiral central charge $c_- = 1/2$. The computation is done with $J_x = J_y = J_z = 1, \kappa = 0.1$.
    (d) Extrapolating $\zeta_{n,a}$ as a function of $1/L_y^2$.
    (e) Phase of partial rotation $\mathcal{T}_1(2\pi /n )$ as it goes through a phase transition from the Ising phase to the toric code phase. $J_z = 2$ line marks the phase transition, and the horizontal dashed (dotted) lines are CFT predictions for Ising (toric code) phase. The computation is done with $J_x = J_y = 1$, $\kappa = 0.1$, $L_x = 100, L_y = 120$.
    }
    \label{fig:Kitaev_zeta}    
\end{figure}

The Kitaev honeycomb model supports two types of topological orders~\cite{Kitaevanyons}: the toric code phase with $c_- = 0$ and the Ising topological order with $c_-=1/2$. In this appendix, we show the higher central charges of these two phases, and see how they interpolate at finite $L_y$.

The toric code phase is obtained when one of $J_x, J_y, J_z$ is much larger than other e.g. $|J_x| > |J_y| + |J_z|$. Perturbation theory calculation confirms the effective theory for this corresponds to Wen's plaquette model~\cite{wen_plaquette}.
Let us now consider the phase with $|J_z| > |J_y| + |J_x|$. With the zigzag boundary condition given in Fig.~\ref{fig:kitaev_main_text}, this corresponds to a edge critical theory between the $m$ and $e$ type boundaries. However, since all topological sectors of toric code has trivial $\zeta_{n, a}$, this does not affect the extraction of partial rotation. Therefore, we will simply reuse the formula Eq.\eqref{eq:t1_from_majorana_spectrum} and its $\sigma$ phase counterpart to compute partial rotation, even though it no longer corresponds to the trivial or $\sigma$ sector of the Ising topological order.

In Fig.~\ref{fig:Kitaev_zeta}(a)-(d), we show the convergence of higher central charge with circumference $L_y$, both for the toric code phase and the Ising phase. The toric code phase has four Abelian anyons with topological spins $\theta_1 = 1, \theta_m = 1, \theta_e = 1, \theta_f = -1$, and correspondingly the higher central charge is always real and positive. We see $\zeta_n$ converges quickly to their theoretical values. In Fig.~\ref{fig:Kitaev_zeta} (e), we show the partial rotation $\mathcal{T}_a(2\pi / n)$. At this fixed $L_y$, phase angle smoothly interpolates between the Ising phase and the toric code phase.


\section{Higher central charge of fermionic $\nu=2/3$ FQH state: non-chiral example}
In this appendix, we numerically evaluate the partial rotation on a simple example of non-chiral topological order; $\nu=2/3$ fractional quantum Hall (FQH) state. This is a fermionic topological order described by the $\U_{1}\times \U_{-3}$ Chern-Simons theory and carries $c_-=0$. We will see that the numerical results for the partial rotations with $n=2,3,4$ has an excellent agreement with the CFT predictions, supporting the validity of the cut-and-glue approach to evaluate the partial rotation $\mathcal{T}_a(2\pi/n)$ even for non-chiral phases.

\subsection{Higher central charge and partial rotation on Abelian fermionic topological order}
Since we are interested in fermionic FQH states instead of bosonic states studied in the main text,  let us first compute the partial rotation of fermionic topologically ordered states and derive the formula for $\mathcal{T}_a(2\pi/n)$.

For simplicity, let us restrict ourselves to the case of Abelian fermionic topological order. In that case, the anyons are described by a super-modular category $\mathcal{C}$ with the form of
\begin{align}
    \mathcal{C} = \mathcal{C}_0 \boxtimes \{1,\psi\}
\end{align}
where $\mathcal{C}_0$ is a modular tensor category, and the operation $\boxtimes$ is Deligne product that physically means stacking two theories.
We show in the leading order that
\begin{align}
\mathcal{T}_a\left(\frac{2\pi}{n}\right) \propto
    \begin{cases}
        e^{\frac{2\pi i}{n} h_a-2\pi i (\frac{2}{n}+n)c_-} \cdot \sum_{b\in\mathcal{C}_0} S_{ab}d_b\theta_b^n & \text{when $n$ is even}, \\
        e^{\frac{2\pi i}{n} h_a-2\pi i (\frac{2}{n}+n)c_-} \cdot \sum_v\sum_{b\in\mathcal{C}_0} S_{ab}d_b\theta_b^n S_{bv}e^{\frac{2\pi i}{n}h_v} & \text{when $n$ is odd},
    \end{cases}
    \label{eq:fermionpartialrot}
\end{align}
where $a\in\mathcal{C}$ specifies the topological sector, and $v$ is a vortex of $\Z_2^f$ fermion parity symmetry carrying the smallest spin. If there are multiple labels of $\Z_2^f$ vortices with the smallest spin we sum over such vortices in the expression, which is meant by the sum $\sum_v$. Formally, $v$ is a simple object of the modular extension of the super-modular category $\mathcal{C}$ denoted as $\Breve{\mathcal{C}} = \mathcal{C}\oplus \mathcal{C}_v$. $\Breve{\mathcal{C}}$ is physically regarded as a bosonic theory obtained by gauging $\Z_2^f$ symmetry of the fermionic theory $\mathcal{C}$, where $\mathcal{C}_v$ describes the defects carrying the vortex of $\Z_2^f$ symmetry.
The above formula can be obtained by following the logic of the analytic computation described in the main text. Namely, by employing the cut-and-glue approach, we first write the partial rotation in terms of the CFT character as
\begin{align}
\begin{split}
    \mathcal{T}_a\left(\frac{2\pi}{n}\right) &= \frac{\mathrm{Tr}_a[e^{i{P}_l\frac{L}{n}}e^{-\frac{\xi_l}{v}H_l}]\mathrm{Tr}_a[e^{i{P}_r\frac{L}{n}}e^{-\frac{\xi_r}{v}H_r}]}{\mathrm{Tr}_a[e^{-\frac{\xi_l}{v}H_l}]\mathrm{Tr}_a[e^{-\frac{\xi_r}{v}H_r}]} \\
    &= \frac{\chi_a(\frac{i\xi_l}{L}+\frac{1}{n}; [\mathrm{AP, AP}]) \chi_a(\frac{i\xi_r}{L}-\frac{1}{n}; [\mathrm{AP, AP}]) }{\chi_a(\frac{i\xi_l}{L}; [\mathrm{AP, AP}])\chi_a(\frac{i\xi_r}{L}; [\mathrm{AP, AP}])}~,
    \label{eq:rotasfermionicCFT}
    \end{split}
\end{align}
where [AP,AP] denotes the boundary condition of $\Z_2^f$ symmetry on a torus, labeled by anti-periodic (AP) or periodic (P) in each cycle of the torus. Due to $L\ll \xi_l$, the contribution from the left edge can be approximated as
\begin{align}
\begin{split}
    \mathcal{T}_a\left(\frac{2\pi}{n}\right) 
    &= e^{\frac{2\pi i}{n}(h_a-\frac{c_-}{24})} \frac{\chi_a(\frac{i\xi_r}{L}-\frac{1}{n}; [\mathrm{AP, AP}]) }{\chi_a(\frac{i\xi_r}{L}; [\mathrm{AP, AP}])}~.
    \end{split}
\end{align}
The denominator does not contribute to the phase of $\mathcal{T}_a\left(\frac{2\pi}{n}\right)$ since
\begin{align}
    \begin{split}
        \chi_a\left(\frac{i\xi_r}{L}; [\mathrm{AP, AP}]\right) &= \sum_b S_{ab} \chi_b\left(\frac{iL}{\xi_r}; [\mathrm{AP, AP}]\right) \approx \frac{d_a}{\mathcal{D}}e^{\frac{2\pi L}{\xi_r}\frac{c_-}{24}},
    \end{split}
\end{align}
so we have
\begin{align}
\begin{split}
    \mathcal{T}_a\left(\frac{2\pi}{n}\right) 
    &= e^{\frac{2\pi i}{n}(h_a-\frac{c_-}{24})} {\chi_a\left(\frac{i\xi_r}{L}-\frac{1}{n}; [\mathrm{AP, AP}]\right) }~.
    \end{split}
\end{align}
In the following let us derive Eq.~\eqref{eq:fermionpartialrot} by cases of even or odd $n$.
\begin{itemize}
    \item Even $n$. We perform the $ST^nS$ modular transformation on the CFT character, which transforms the spin structure of the torus as $ST^nS:$ [AP,AP]$\to$[AP,AP] for even $n$. We can then approximate the CFT character by the contribution from the vacuum as
    \begin{align}
        \begin{split}
            \chi_a\left(\frac{i\xi_r}{L}-\frac{1}{n}; [\mathrm{AP, AP}]\right)  &= \sum_b (ST^n S)_{ab} \chi_b\left(\frac{iL}{n^2\xi_r}+\frac{1}{n}; [\mathrm{AP, AP}]\right) \\
            &\approx (ST^n S)_{a,1}e^{-\frac{2\pi i}{n}\frac{c_-}{24}} e^{\frac{2\pi L}{n^2\xi_r}\frac{c_-}{24}} \\
        &= \frac{1}{\mathcal{D}} e^{-2\pi i(n + \frac{1}{n})\frac{c_-}{24}} e^{\frac{2\pi L}{n^2\xi_r}\frac{c_-}{24}}\sum_{b\in\mathcal{C}}S_{ab}d_b\theta_b^n \\
        &\propto e^{-2\pi i(n + \frac{1}{n})\frac{c_-}{24}} e^{\frac{2\pi L}{n^2\xi_r}\frac{c_-}{24}}\sum_{b\in\mathcal{C}_0}S_{ab}d_b\theta_b^n.
    \end{split}
    \end{align}
    In the last equation, we used that the transparent fermion $\psi$ braids trivially with other anyons and has $S_{a\psi}\theta_\psi^n\propto 1$ for even $n$. We can then obtain the expression in Eq.~\eqref{eq:fermionpartialrot}.
    
    \item Odd $n$. We again perform the $ST^nS$ modular transformation on the CFT character, which now transforms the spin structure of the torus as $ST^nS:$ [AP,AP]$\to$[P,AP] for odd $n$. We then have
    \begin{align}
        \begin{split}
            \chi_a\left(\frac{i\xi_r}{L}-\frac{1}{n}; [\mathrm{AP, AP}]\right)  &= \sum_b (ST^n S)_{ab} \chi_b\left(\frac{iL}{n^2\xi_r}+\frac{1}{n}; [\mathrm{P, AP}]\right) \\
            &=\sum_{v\in\mathcal{C}_v} (ST^n S)_{a,v}e^{\frac{2\pi i}{n}(h_v-\frac{c_-}{24})} e^{-\frac{2\pi L}{n^2\xi_r}(h_v-\frac{c_-}{24})} \\
        &\propto e^{-2\pi i(n + \frac{1}{n})\frac{c_-}{24}}\sum_{v\in\mathcal{C}_v} e^{-\frac{2\pi L}{n^2\xi_r}h_v}\sum_{b\in\mathcal{C}_0}S_{ab}d_b\theta_b^n S_{bv}e^{\frac{2\pi i}{n}h_v}.
        \end{split}
    \end{align}
\end{itemize}
The expression now involves the sum over the defects $v$ carrying the $\Z_2^f$ vortex, which labels the states in the periodic sector $\mathcal{C}_v$. Due to the exponentially dropping factor $e^{-\frac{2\pi L}{n^2\xi_r}h_v}$, the sectors carrying the lowest spin give the leading contributions. We then arrive at the expression in Eq.~\eqref{eq:fermionpartialrot} to the leading order.

\subsection{$\nu=2/3$ FQH state}
Here let us consider the partial rotation on the $\nu=2/3$ FQH state, described by $\U_{1}\times \U_{-3}$ Chern-Simons theory. Since this theory is Abelian, it allows the expression in terms of $\mathcal{C}_0\boxtimes\{1,\psi\}$ with a modular category $\mathcal{C}_0$; we have $\mathcal{C}_0=\{1,p,p^2\}$ with $p$ an Abelian anyon carrying spin $h_p=1/3$. The partial rotation $\mathcal{T}_1(2\pi/n)$ with small even $n$ are computed as $\mathcal{T}_1(2\pi/2)\propto -i, \mathcal{T}_1(2\pi/4)\propto i$.

To compute the partial rotation for odd $n$, we need to look at the periodic sector $\mathcal{C}_v$ carrying the $\Z_2^f$ vortex. The modular extension $\Breve{\mathcal{C}}$ of the $\nu=2/3$ state is described by $\U_4\times \U_{-12}$ Chern-Simons theory with the diagonal fermion $([2]_4, [6]_{12})$ condensed. Here, we label the anyon of $\U_{\pm m}$ theory as $[j]_m$ with $j\in\Z_m$. There are six simple objects in $\mathcal{C}_v$ carrying the $\Z_2^f$ vortex, given by $v_j := ([1]_4, [2j-1]_{12})$ with $1\le j\le 6$. The spins are $h_{v_1}=h_{v_5}=h_{v_7}=h_{v_{11}}=1/12, h_{v_3}=h_{v_9}=3/4$.

Let us study the simplest case at $n=1$. We have $\sum_b \theta_b S_{b,v} = \sum_b\theta_{b\times v^*}/\theta_{v^*} = \theta_v^{-1}$, so we get $\mathcal{T}_1(2\pi)=1$.
We then consider the partial rotation at $n=3$, $\mathcal{T}_1(2\pi/3)$. For $v=v_1, v_5, v_7, v_{11}$, one can see that $\sum_b \theta_b^3 S_{b,v}$ gives zero, so in the leading order the partial rotation vanishes.
Then, the main contribution comes from the subleading order with $v=v_3,v_9$. In that case we have $\sum_b \theta_b^3 S_{b,v}\propto 1$, so we get $\mathcal{T}_1(2\pi/3)=e^{\frac{2\pi i}{3}h_v} = i$. We summarize the results of the partial rotation for small $n$ in Table~\ref{tab:nu=2/3}. 
\begin{table}[h!]
\renewcommand*{\arraystretch}{1.5}
\centering
\begin{tabular}{c|c}
\hline
       &  $\mathcal{T}_1\left(\frac{2\pi}{n}\right)$ \\ \hline
\multirow{1}{*}{$\nu=2/3$}     &  $1,-i,i,i$  \\ \hline

\end{tabular}
\caption{The phases of the partial rotation $\mathcal{T}_1(\frac{2\pi}{n})$ for $n=1,2,3,4$ in $\nu=2/3$ state.}
\label{tab:nu=2/3}
\end{table}

Following the methodology detailed in the main text, we numerically evaluate the partial rotation of a $\nu = 2/3$ filled LLL of fermions, with model interaction $V_1 = 1$ plus a small perturbation $\delta V_3 = 0.1$. From the orbital space entanglement spectrum, we confirm that the obtained ground state is described by $\U_{1}\times \U_{-3}$ Chern-Simons theory. The result of $\mathcal{T}_1(\frac{2\pi}{n})$ in the
trivial sector is presented in Fig.~\ref{fig:nonchiral}, which converges to the expected phase shown in Table~\ref{tab:nu=2/3}.

\begin{figure}[htbp]
    \centering
    \includegraphics[width = 0.25\textwidth]{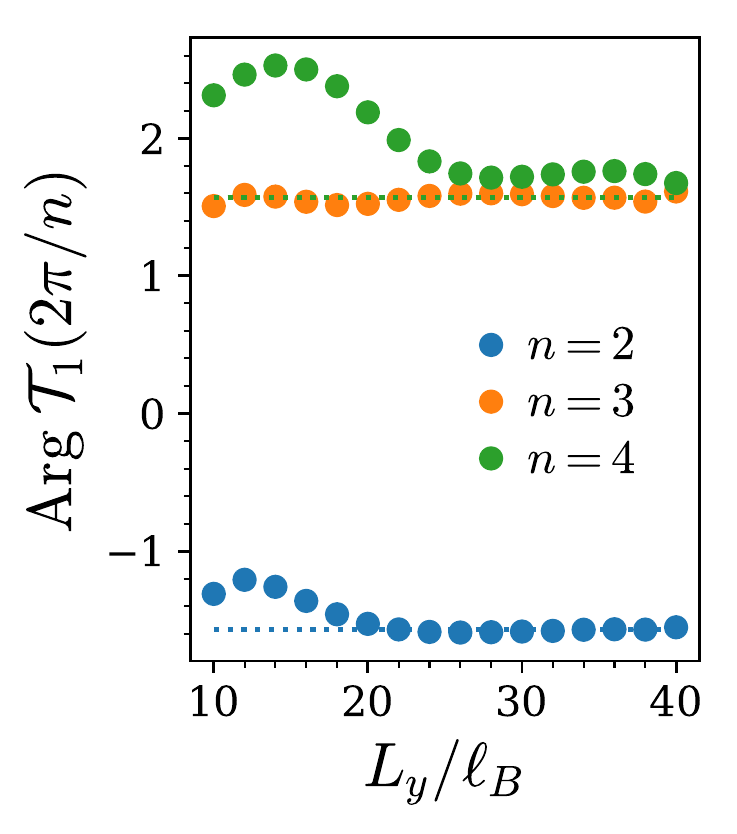}
    \caption{Partial rotation $\operatorname{Arg} \: \mathcal{T}_1(2\pi/n)$ of the fermionic $\nu = 2/3$ FQH state. The dotted lines are the CFT predictions given in Table~\ref{tab:nu=2/3}.}
    \label{fig:nonchiral}    
\end{figure}

\subsection{Partial rotation of Abelian fermionic states as obstructions to gapped edge}
Finally, let us comment that the partial rotation of an Abelian fermionic topological order can also define an obstruction to gapped edge state.
When the fermionic theory $\mathcal{C}=\mathcal{C}_0\boxtimes\{1,\psi\}$ admits a gapped edge, it can be shown that the modular category $\mathcal{C}_0$ must have a Lagrangian subgroup, i.e., a bosonic theory described by $\mathcal{C}_0$ must admit a bosonic gapped boundary~\cite{Kobayashi2022FQH}. This implies that the higher central charge of $\mathcal{C}_0$
\begin{align}
    \zeta_n(\mathcal{C}_0) = \frac{\sum_{a\in\mathcal{C}_0} \theta_a^n}{|\sum_{a\in\mathcal{C}_0} \theta_a^n|} \quad \text{with $\gcd(n,\frac{N_{\mathrm{FS}}(\mathcal{C}_0)}{\gcd(n,N_{\mathrm{FS}}(\mathcal{C}_0) )})=1$}
\end{align}
define obstructions to gapped boundary of an Abelian fermionic theory. So, comparing with the result of the partial rotation described above, $\mathcal{T}_1(2\pi/n)$ defines an obstruction to a gapped edge for even $n$ satisfying $\gcd(n,\frac{N_{\mathrm{FS}}(\mathcal{C}_0)}{\gcd(n,N_{\mathrm{FS}}(\mathcal{C}_0) )})=1$.
For the $\nu=2/3$ state described above, the partial rotation $\mathcal{T}_1(2\pi/n)$ with $n=2,4$ becomes an obstruction to gapped edge.

\section{Evaluating the action of partial rotation of an MPS wave function}

Here we explain how to evaluate the action of partial rotation of an MPS wave function, with quantum Hall systems as an example. In the following we will first review the quantum Hall DMRG setup developed in Ref.~\onlinecite{exactmps,FQHEDMRG,BosonicFQHEDMRG}. We consider an infinite cylinder geometry with $y$ running around the circumference, and $x$ running along the infinite direction (see Fig.~\ref{fig:cylinder}). We work in the Landau gauge and the corresponding LLL orbital takes the form (see Fig.~\ref{fig:Laughlin} (a))
\begin{equation}
    \varphi_n(x, y)=\frac{e^{i k_n y-\frac{1}{2 \ell_B^2}\left(x-k_y^n \ell_B^2\right)^2}}{\sqrt{L_y\ell_B \pi^{1 / 2}}}, \quad k_y^n=\frac{2 \pi n}{L_y}
\end{equation}
where $n \in \mathbb{Z}$ and orbital $n$ localizes at $k_n \ell_B^2 = 2 \pi n \ell_B / L_y$. We note that in order to perform the partial rotation, it is important to work in an eigenbasis in momentum $k_y$, which implicitly assumes translation symmetry, continuous or discrete, along the $y$ direction. For a lattice model or a $k$-space continuum model, such a basis can be obtained by hybrid Wannier localization in the $y$ direction, which maps the original model to a cylinder in mixed $(x, k_y)$ basis \cite{Tomo_2020_wannierization}. 

There is a natural one-dimensional iMPS representation of quantum Hall wavefunction in this basis,
\begin{equation}
    |\Psi[B]\rangle=\sum_{\left\{j_n\right\}}\left[\cdots B^{[0] j_0} B^{[1] j_1} \cdots\right]\left|\ldots, j_0, j_1, \ldots\right\rangle
\end{equation}
where $B^{[n] j_n}$ are $\chi \times \chi$ matrices and $\left|j_n\right\rangle, j_n \in\{0,1,\cdots,N_{\mathrm{boson}}\}$ represent the boson occupancy at orbital $n$. Both bond dimension $\chi$ and onsite boson number cutoff $N_{\mathrm{boson}}$ constrain the representability of the iMPS. However, as we will show in later sections of the Supplemental Materials, the partial rotation $\mathcal{T}_a(2\pi/n)$ quickly saturates to its true value at relatively small $\chi$ and $N_{\mathrm{boson}}$. 

One key ingredient to accelerate the DMRG algorithm and allow the action of partial rotation is to implement both particle number and momentum conservation,
\begin{equation}
    \hat{C} =\sum_n \hat{C}_n \equiv \sum_n\left ( \hat{N}_n-\frac{p}{q} \right ), \quad \hat{K} =\sum_n \hat{K}_n \equiv \sum_n n\left ( \hat{N}_n-\frac{p}{q} \right )
    \label{eq:QHMPS_conservation}
\end{equation}
where $\hat{N}_n$ is the number operator at site $n$ and $\nu = p/q$. Now for a generic conserved $\mathrm{U}(1)$ charge 
$\hat{Q} = \hat{C}$ or $\hat{K}$,
we can decompose the symmetry operator $\hat{Q}$ into $\hat{Q} = \sum_{n<\bar{n}} \hat{Q}_n+\sum_{n>\bar{n}} \hat{Q}_n = \hat{Q}_{\mathrm{A}} + \hat{Q}_{\mathrm{B}}$, where $\bar{n}$ is the auxiliary bond across A and B. For the moment we will assume A, B represent two sets of LLL orbitals and discuss the difference between orbital-space cut and real-space cut later. Then for a Schimidt decomposition across bond $\bar{n}$, $\ket{\Psi} = \sum_{\alpha} \lambda_{\alpha}^2 \ket{\alpha}_\mathrm{A} \otimes \ket{\alpha}_\mathrm{B}$, we can require the Schimidt states to be eigenstates of $\hat{Q}_{\mathrm{A/B}}$,
\begin{equation}
    \hat{Q}_\mathrm{B} \ket{\alpha}_\mathrm{B} = \bar{Q}_{\bar{n};\alpha}  \ket{\alpha}_\mathrm{B}
\end{equation}
At the level of MPS, local charge conservation on site $n$ takes the form,
\begin{equation} \label{eq:momentum_conservation}
    \left[\bar{Q}_{\bar{n} ; \beta}-\bar{Q}_{\bar{n} + 1 ; \alpha}-\hat{Q}_{n ; j}\right] B_{\alpha \beta}^{[n] j}=0
\end{equation}
where $\bar{n} < n < \bar{n} + 1$. This conservation law naturally generalizes to the corresponding symmetry action. For example, pictorially, momentum conservation pictorially means
\begin{equation}
\includegraphics[scale = 1.3]{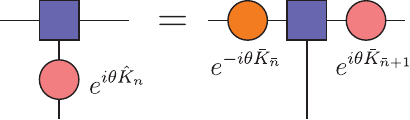}
\end{equation}
where $\hat{U} = e^{i \theta \hat{K}}$ is the symmetry action of rotating site $n$ along the $y$ direction by angle $\theta$.

\subsection{Partially rotate an MPS}

Now we explain how to exploit the implemented momentum conservation to perform the partial rotation. Naively, $\hat{K}_\mathrm{B} = \sum_{n>\bar{n}} \hat{K}_n$ has the interpretation of total momentum to the right of bond $\bar{n}$, so one would expect the action of partial rotation of angle $\theta$ to be
\begin{equation}
    e^{- i \theta \hat{K}_\mathrm{B}} \ket{\Psi} = e^{- i \theta \hat{K}_\mathrm{B}} \sum_{\alpha} \lambda_{\alpha} \ket{\alpha}_\mathrm{A} \otimes \ket{\alpha}_\mathrm{B} = \sum_{\alpha} \lambda_{\alpha}^2 e^{- i \theta \bar{K}_{\bar{n};\alpha}} \ket{\alpha}_\mathrm{A} \otimes \ket{\alpha}_\mathrm{B}
\end{equation}
This would be obvious for a finite MPS whose momentum eigenvalues $\bar{K}_{\bar{n};\alpha}$ are finite. However, for iMPS, a more rigorous proof presented below would be appreciated.

\begin{thm}
For an infinite MPS $\ket{\Psi}$, the expectation value $\mathcal{T}_a(\theta)$ of partial rotating all sites in B by angle $\theta$ is equivalent to inserting a $e^{- i \theta \bar{K}_{\bar{n};\alpha}}$ factor to the bond $\bar{n}$ across A, B.
\begin{equation}
    \bra{\Psi} \prod_{n > \bar{n}} e^{i \theta \hat{K}_n} \ket{\Psi} = \sum_{\alpha} \lambda_{\alpha}^2 \bra{\alpha}_\mathrm{A} \otimes \bra{\alpha}_\mathrm{B} e^{- i \theta \bar{K}_{\bar{n};\alpha}} \ket{\alpha}_\mathrm{A} \otimes \ket{\alpha}_\mathrm{B} = \sum_{\alpha} \lambda_{\alpha}^2 e^{- i \theta \bar{K}_{\bar{n};\alpha}}
\end{equation}
Pictorially,
\begin{equation}
\includegraphics[scale = 1.3]{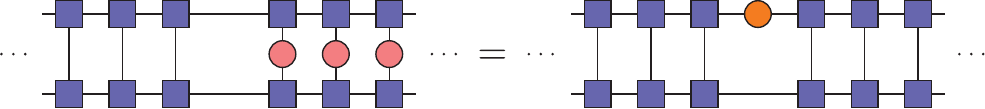}
\end{equation}
\end{thm}

\begin{proof}
To prove this theorem, we introduce the formulation of transfer matrices. For simplicity, we assume the unit cell size to be one in the following proof. The transfer matrix $M_U$ is defined to be $M_U = \sum_{j,j'} U_{j,j'} B^j \otimes \bar{B}^{j'}$. The most important property of transfer matrix $M_1$ (with no operator insertion) is that we can approximate $M_1^N$ by its dominant eigenvector $\rho$ with eigenvalue $\lambda = 1$,
\begin{equation} \label{eq:transfer_matrix}
\includegraphics[scale = 1.2]{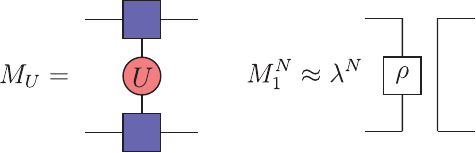}
\end{equation}
since all eigenvectors with eigenvalue $\lambda' < 1$ would be negligible in the iMPS limit $N \to \infty$. The theorem can be rewritten in a more precise form using the language of transfer matrices,
\begin{equation}
    M_U^N \approx e^{i \phi} \bar{U}^{-1} M_1^N, \quad \text{where } U = e^{i \theta \hat{K}_n}, \; \bar{U}^{-1} = e^{- i \theta \bar{K}_{\bar{n};\alpha}} \rho_{\alpha,\alpha'}
\end{equation}
where $\phi$ is an overall phase ambiguity. Using the momentum conservation in Eq.~\eqref{eq:momentum_conservation}, we can show $e^{- i \theta \bar{K}_{\bar{n};\alpha}} \rho_{\alpha,\alpha'}$ is an eigenvector of $M_U$ with eigenvalue $\lambda = 1$, which therefore must be the dominant eigenvector of $M_U$,
\begin{equation}
\includegraphics[scale = 1.2]{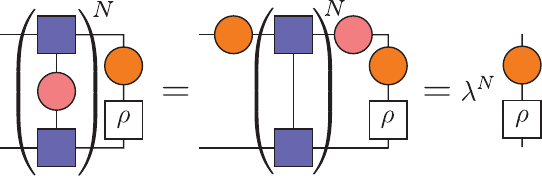}
\end{equation}
We note that $e^{i \phi} e^{- i \theta \bar{K}_{\bar{n};\alpha}} \rho_{\alpha,\alpha'}$ is also a dominant eigenvector, so there is an overall phase ambiguity that cannot be determined by this argument. We will see the consequence of this ambiguity in the next section. Combined with the property of transfer matrix in Eq.~\eqref{eq:transfer_matrix}, we have established $M_U^N \approx e^{i \phi} \bar{U}^{-1} M_1^N$ and therefore proved the theorem. 
\end{proof}

\subsection{Matching iMPS momentum labels to CFT momentums}

The most subtle but important aspect in the evaluation of partial rotation is to resolve some of the gauge degrees of freedom
in matching the momentum labels $\bar{K}_{\alpha}$ in the iMPS representation to the CFT momentum labels of a state.
The conserved-charges (i.e., charge and momentum) labels for particles on each site are given by Eq.~\eqref{eq:QHMPS_conservation}.  The conserved-charge labels for the Schmidt states across any bond is the accumulations of the those to the left of the cut; which is sensitive even to electrons and dipole that are far away from the entanglement cut.
As we ``grow'' the MPS chain in the iDMRG algorithm, such electrons/dipoles may be ``stuck'' based on the initialization of the DMRG.
For example, an electron far away may shift the conserved-charge labels via $\hat{C} \to \hat{C} + q$ and $\hat{K} \to \hat{K} + d_\text{electron} \hat{C}$, where $d_\text{electron}$ is the distance (proportional to number of DMRG steps) of the charge.
Therefore the momentum label $\bar{K}_{\bar{n};\alpha}$ has two integer ambiguity $x_0$ and $x_1$ to the determined,
\begin{equation}
    \bar{K}_{\alpha} \mapsto \bar{K}_{\alpha}' = \bar{K}_{\alpha} + x_1 \bar{C}_{\alpha} + x_0 .
\end{equation}
To reproduce the CFT calculation in Eq.~\eqref{eq:mainresult}, we can match the entanglement spectrum and the (1+1)D edge CFT spectrum to fix the momentum labels.
We demand, for the ground state corresponding to the vacuum, that (a) the highest weight state (corresponding to the lowest energy state of the entanglement Hamiltonian) has zero momentum and charge, and (b) that sectors with opposite charges have the same momentum labels.
(For the semion ground state, its entanglement spectrum has two highest weight state--we assign them to have charges $\pm 1/2$ and zero momentum.)

\subsection{Orbital cut v.s. real-space cut}

Now we are ready to talk about the difference between the orbital cut and real-space cut. The CFT prediction in Eq.~\eqref{eq:mainresult} assumes a real space partition. Before we explain the numerical realization, we first discuss the physical implication of a real-space cut at different locations. In quantum Hall systems, the entanglement spectrum depends on where the real-space cut is made.
In particular, we want to make the cut such that anyon $a$ can appear right on the cut, then the entanglement spectrum would agree best with the 1D edge CFT spectrum. For fermions, this happens when the cut is in the middle of two LLL orbital centers; for bosons, this happens when the cut is right at the LLL orbital center.

This is most obvious in the thin cylinder limit, where the $\nu = 1/2$ bosonic Laughlin state takes the form of
\begin{equation}
    \ket{\Psi} = \ket{\cdots 0 1 0 1 0 1\cdots}
\end{equation}
in the occupation basis. This is usually referred as the root configuration of the state (or pattern of zeros).
For the entanglement spectrum corresponding to the vacuum sector, we want to choose the boundaries of the unit cell such that the electron (i.e., the `1' of the root configuration) is in the middle of the unit cell.  (This ensures that entanglement spectrum has a unique highest weight state and a symmetry between positive and negative charge selection sectors.)
The real-space cut corresponding to boundaries of the unit cell.

Therefore, cutting through the $0$'s gives the trivial sector and cutting through the $1$'s (related by threading one flux quanta through the cylinder) gives the semion sector.
Furthermore, this also explains why for fermionic Laughlin states ($q$ odd), the unit cell cuts between a pair of neighboring sites.
This identification remains true at finite circumference $L$.

Finally, we briefly explain the numerical realization of the real-space cut. We take the most straightforward approach based on the real-space entanglement spectrum (RSES) algorithm developed in Ref.~\onlinecite{exactmps}. For a real-space partition at $x = x_c$, RSES algorithm reconstructs the orbital space MPS $\ket{\Psi[B]}$ obtained in iDMRG into a new one $\ket{\psi[B]_{\mathrm{RSES}}}$ such that a particular bond $\bar{n}_c$ gives the real-space Schmidt decomposition at $x_c$. This is accomplished by splitting each matrix $B^{[n]}$ into two $B^{[n]}_\mathrm{A}$ and $B^{[n]}_\mathrm{B}$ based on the weight of the orbital $\varphi_n$ supported on A, B, and then swap all $B^{[n]}_\mathrm{A}$ to the left of bond $\bar{n}_c$ and $B^{[n]}_\mathrm{B}$ to the right. Then we can easily evaluate the action of partial rotation with the techniques discussed above. We note that the difference between orbital cut and real-space cut is quite general in systems where band projection was used, since the Wannier function usually span across a few sites.


\section{Extracting higher central charge of the $\nu=1/2$ bosonic Laughlin state}

\subsection{Extracting higher central charge from scratch}

\begin{figure}[htbp]
    \centering
    \includegraphics[width = 0.5\textwidth]{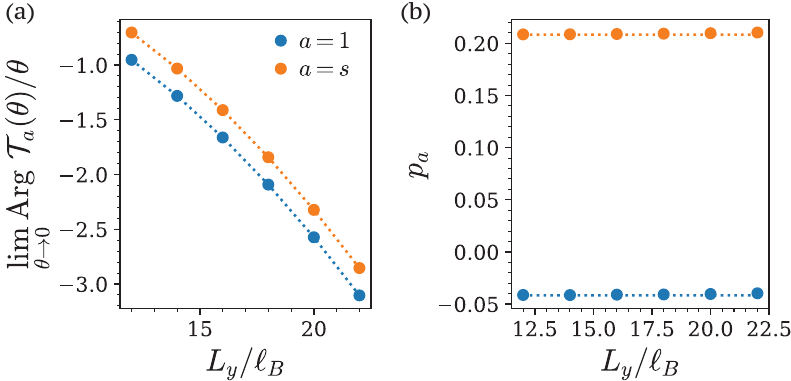}
    \caption{(a) Infinitesimal rotation $\lim_{\theta \to 0} \: \operatorname{Arg} \: \mathcal{T}_a(2\pi/n)$ of the $\nu = 1/2$ bosonic Laughlin state extracted using Eq.~\eqref{eq:schmidt}. (b) The corresponding momentum polarization $p_a = h_a - c_-/24$ extracted using Eq.~\eqref{eq:momentumpolarization}. The dotted lines are the CFT predictions.}
    \label{fig:DMRG_mp}    
\end{figure}

In the main text Fig.~\ref{fig:Laughlin}, we only present the numerical result of $\mathcal{T}_a(2\pi/n)$. We note that $\mathcal{T}_a(2\pi/n)$ itself is sufficient in determining the gappability of the edge, without referring to the actual value of the higher central charge $\zeta_n$. However, if one wants to extract $\zeta_n$, one also needs to know about the chiral central charge $c_-$, which could be obtained from either some prior knowledge about the topological order, or the momentum polarization \cite{Qi2012momentumpolarization,FQHEDMRG}, or certain entanglement measures \cite{Kim2022cminus, Kim2022modular, Fan2022cminus, Fan2022QHE}. Our numerical algorithm can easily reproduce the momentum polarization by taking the large $n$ limit, though the physics is completely different as explained in the main text. For completeness, we show the momentum polarization, and the resulting higher central charge in this section. 

In quantum Hall systems, the momentum polarization takes the universal form \cite{FQHEDMRG}
\begin{equation} \label{eq:mp_qh}
    \lim_{n \to \infty} \mathcal{T}_a\left(\frac{2\pi}{n}\right) = \exp \left[\frac{2 \pi i }{n} \left(h_a-\frac{c_-}{24}-\frac{\eta_H}{2 \pi \hbar} L^2\right)\right]
\end{equation}
where $\eta_H = \hbar \nu \mathcal{S}/ 8 \pi \ell_B^2$ is the universal Hall viscosity with $\mathcal{S}$ being the topological shift \cite{Read_2009_shift}. We first evaluate the action of infinitesimal rotation and then fit the result to Eq.~\eqref{eq:mp_qh} (see Fig.~\ref{fig:DMRG_mp} (a)), where we find $\mathcal{S} = 1.998$ for both topological sectors, in excellent agreement of the CFT prediction $\mathcal{S} = 2$. Subtracting off the the Hall viscosity contribution, we plot the resulting momentum polarization $p_a = h_a - c_-/24$ in Fig.~\ref{fig:DMRG_mp} (b), which again gives the expected $c_- = 1$ and $h_s = 1/4$. We note that we can only reach $L_y = 22 \ell_B$ using the largest practical onsite boson cutoff $N_{\mathrm{boson}}$, whereas for partial rotation we can easily reach $L_y > 40 \ell_B$. This is because a much higher precision is required in evaluating $\mathcal{T}_a\left(\theta\right)$ in order to reliably extract $p_a$, a point we will return to in the next section.

\begin{figure}[htbp]
    \centering
    \includegraphics[width = 0.5\textwidth]{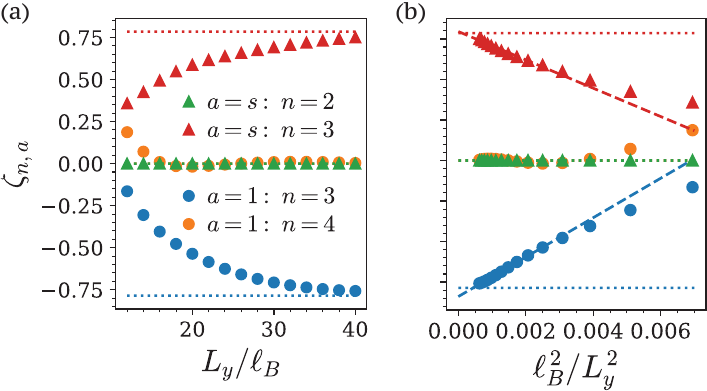}
    \caption{(a) (Twisted) higher central charge $\zeta_{n,a}$ extracted from $\mathcal{T}_a(2\pi/n)$ in Fig.~\ref{fig:Laughlin} (b) and chiral central charge $c_- = 1$ in Fig.~\ref{fig:DMRG_mp} using the main result Eq.~\eqref{eq:mainresult}. The dotted lines are the CFT predictions as usual. (b) Extrapolating $\zeta_{n,a}$ as a function of $\ell_B^2/L_y^2$.}
    \label{fig:zeta}    
\end{figure}

Using $c_- = 1$ and $h_s = 1/4$ extracted from momentum polarization, we can extract the (twisted) higher central charge $\zeta_{n,a}$ using the main result Eq.~\eqref{eq:mainresult}. As shown in Fig.~\ref{fig:zeta}, $\zeta_{n,a}$ all converge to the expected values as shown in Table \ref{tab:phases} at sufficiently large $L_y$.

\subsection{Numerical performance of partial rotation}

\begin{figure}[htbp]
    \centering
    \includegraphics[width = 0.5\textwidth]{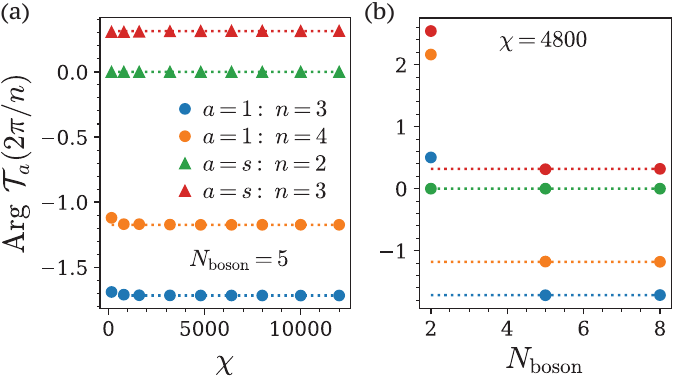}
    \caption{Partial rotation $\mathcal{T}_a(2\pi/n)$ as a function of bond dimension $\chi$ and onsite boson number cutoff $N_{\mathrm{Boson}}$ at $L_y = 40 \ell_B$. $\mathcal{T}_a(2\pi/n)$ saturates at relatively small $\chi = 1600$ and $N_{\mathrm{Boson}} = 5$.}
    \label{fig:performance}    
\end{figure}

Now we briefly discuss the numerical performance of our protocol to extract higher central charge from partially rotating an MPS wavefunction. Compared to free fermion models, e.g. the Kitaev model, the MPS representation of a topological state wavefunction is generically not exact. Both the bond dimension $\chi$ and the onsite boson number cutoff $N_{\mathrm{Boson}}$ put constraints on the representability of the MPS. However, the quantity we are interested in, i.e. the partial rotation $\mathcal{T}_a(2\pi/n)$ rapidly saturates to its true value up to the largest system size we reach (see Fig.~\ref{fig:performance}). In practice, for a relatively small $\chi$ and $N_{\mathrm{Boson}}$, we can already obtain $\mathcal{T}_a(2\pi/n)$ and therefore $\zeta_{n,a}$ to rather high precision.

For readers who are familiar with the momentum polarization, we note that the precision required in the evaluation of $\mathcal{T}_a(\theta)$ is much higher in order to reliably extract the momentum polarization $p_a$ due to the Hall viscosity term that is quadratic in cylinder circumference $L_y$. Specifically, the error in $\zeta_{n,a}$ and $p_a$ are related to the error in $\operatorname{Arg}\mathcal{T}_a(\theta)$ by
\begin{equation}
    \Delta \zeta_{n,a} \sim \frac{\Delta \operatorname{Arg} \mathcal{T}_a(\theta)}{\operatorname{Arg} \mathcal{T}_a(\theta)}, \quad \Delta p_a \sim L_y^2 \frac{\Delta \operatorname{Arg} \mathcal{T}_a(\theta)}{\operatorname{Arg} \mathcal{T}_a(\theta)}
\end{equation}
We note that $\zeta_{n,a}$ and $p_a$ are both $O(1)$, but $L_y^2$ can be fairly large to avoid finite size effect. Therefore, a much smaller $\chi$ and $N_{\mathrm{Boson}}$ is required to reliable extract the higher central charge compared to that one would use for momentum polarization.
 
The main numerical challenge in our protocol to extract the higher central charge is the interpolating between higher central charge and momentum polarization discussed in an earlier section. We need to reach system size $L_y > n^2 \xi_r$ to reliably extract higher central charge $\zeta_{n,a}$. In quantum Hall systems, by comparing Eq.~\eqref{eq:momentumpolarization} and Eq.~\eqref{eq:mp_qh}, we find
\begin{equation}
    \xi_r=\sqrt{\frac{2 \pi^2}{3}} \frac{c_{-}}{\nu \mathcal{S}} \ell_B = \sqrt{\frac{2}{3}} \pi \ell_B
\end{equation}
where the last equality assumes the $\nu = 1/2$ bosonic Laughlin state. In order to obtain $\zeta_{n,a}$ with $n \geq 5$, we need to reach $L_y \geq 25 \sqrt{\frac{2}{3}} \pi \ell_B \approx 64 \ell_B$, which can be challenging in practice.



\vfill



\end{document}